\documentclass[submission,copyright,creativecommons]{eptcs}
\providecommand{\event}{GCM 2019} 

\usepackage[utf8]{inputenc}

\usepackage{graphicx}

\usepackage{amsmath}
\usepackage{amssymb}

\usepackage{amsfonts}
\usepackage{extarrows}
\usepackage{amstext}
\usepackage{stmaryrd}

\usepackage{latexsym}
\usepackage{rotating}
\usepackage{color}
\usepackage{fancybox}
\usepackage{tikz}
\usetikzlibrary{chains,fit,positioning}
\usepackage{url}
\usepackage{dsfont} 
\usepackage{cutwin}
\usepackage{algorithm2e}

\usepackage{styles/tikz-atposition}
\usepackage{styles/tikz-cellular}
\usepackage{styles/tikz-circlesnarrows1}
\usepackage{styles/tikz-colorsnstyles}
\usepackage{styles/tikz-extension}

\usepackage{styles/tikz-derivation}
\usepackage{styles/tikz-vcond} 
\usepackage{wrapfig}
\usepackage{styles/tikz-lok}
\usepackage{styles/tikz-railroad}
\let\nlok=\lok

\providecommand\ignore[1]{}

\usepackage{styles/mytheorem}
\usepackage{styles/abbr-habel} 
\usepackage{styles/abbr-habel-tikz}
\usepackage{styles/abbr-habel-tikz2} 
\usepackage{styles/abbr-poskitt-plump}
\usepackage{styles/abbr-hubatschek-tikz}
\usepackage{styles/abbr-sandmann-tikz}
\usepackage{styles/tikz-railroad-macro}
\usetikzlibrary{snakes,positioning,decorations.pathreplacing,arrows.meta}
\usetikzlibrary{arrows,automata,calc}
\usepackage{tabularx}

\tikzstyle{round corners}=[rounded corners=0.5ex]
\usepackage{float}
\usepackage{graphics}
\usepackage{booktabs}
\def\titlerunning{Rule-based Graph Repair}
\def\authorrunning{Christian Sandmann \&  Annegret Habel}

\newcommand{\Pres}{\mathtt{Pres}}
\newcommand{\pres}{\mathtt{pres}}

\providecommand{\longv}[1]{#1} 
\providecommand{\ap}[1]{} 
\providecommand{\forget}[1]{} 
\renewcommand{\red}{\color{black}}

\begin{document}
\title{Rule-based Graph Repair\thanks{This work is partly supported by the German Research Foundation (DFG), Grants HA 2936/4-2 and TA 2941/3-2 (Meta-Modeling and Graph Grammars: Generating Development Environments for Modeling Languages).}}
\author{Christian Sandmann, Annegret Habel
\institute{Universit\"at Oldenburg \email{\{habel,sandmann\}@informatik.uni-oldenburg.de}}}
\maketitle 



\begin{abstract}Model repair is an essential topic in model-driven engineering. Since models are suitably formalized as graph-like structures, we consider the problem of rule-based graph repair: Given a rule set and a graph constraint, try to construct a graph program based on the given set of rules, such that the application to any graph yields a graph satisfying the graph constraint. 
We show the existence of repair programs for {\red specific} constraints, and show the existence of rule-based repair programs for {\red specific} constraints compatible with the rule set.\end{abstract}

\section{Introduction}

In model-driven software engineering the primary artifacts are models, which have to be consistent w.r.t. a set of constraints (see e.g. \cite{Ehrig-etal15a}). These constraints can be specified by the Object Constraint Language (OCL) \cite{web:OCL24}. To increase the productivity of software development, it is necessary to automatically detect and resolve inconsistencies arising during the development process, called model repair (see, e.g. \cite{Nentwich-etal03a,Macedo-etal17a,Nebras-etal17a}). Since models can be represented as graph-like structures \cite{Biermann-Ermel-Taetzer12a} and a subset of OCL constraints can be represented as graph conditions \cite{Radke+18a,Bergmann14a}, we investigate graph repair and rule-based graph repair.
 
Firstly, the problem of \emph{graph repair} is considered: Given a graph constraint~$d$, we derive repairing sets from the constraint $d$ and try to construct a graph program using this rule set, called \emph{repair program}.
The repair program is constructed, such that the application to any graph yields a graph satisfying the graph constraint. 
Secondly, we consider the problem of \emph{rule-based graph repair}: Given a set of rules $\R$ and a constraint $d$, try to construct a repair program $P$ based on the rule set~$\R$, i.e., we allow to equip the rules of $\R$ with the dangling-edges operator, context, application conditions \cite{Habel-Pennemann09a}, and interface \cite{Pennemann09a}. 

{\bf Rule-based repair problem}
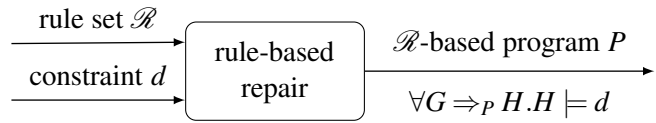
\[\tikz[node distance=6em,inner sep=5pt,rounded corners]{
\node(grammar)[]{};
\node(constraint)[strictly below of=grammar,node distance=1em]{};
\node(rulebased)[strictly right of=constraint,draw=black,minimum height=3em,yshift=1em,node distance=6em]{
\begin{tabular}{c}rule-based\\repair\end{tabular}};
\node(result)[node distance=10em,strictly right of=rulebased]{};
\draw[arrow] (constraint) to node[above]{constraint $d$} ([yshift=-1em]rulebased.west);
\draw[arrow] (grammar) to node[above]{rule set $\R$} ([yshift=1em]rulebased.west);
\draw[arrow] (rulebased) to node[above]{$\R$-based program $P$} node[below]{\begin{tabular}{c}$\forall G \dder_{P} H.H \models d$\end{tabular}} (result);
}\]

If a graph $G$ is generated by a grammar with rule set $\R$, then, after the application of an $\R$-based program, the result can be generated by the grammar, too. {\red This is interesting in contexts where the language is defined by a grammar, like triple graph grammars \cite{Schuerr94b}.}


As main results, we show that, (1) there are repair programs for all ``proper'' conditions, and (2) there are rule-based repair programs for proper conditions provided that the given rule set is compatible with the rule sets of the original program.

We illustrate our approach by a small railroad system.
\begin{example}[railroad system]\label{ex:railroad}The specification of a railroad system is given in terms of graphs, rules (for moving the trains), and conditions. The basic items are waypoints, bi-directional tracks and trains. The static part of the system is given by a directed rail net  graph: tracks are modeled by undirected edges (or a pair of directed edges, respectively) and trains are modeled by edges. Source and target nodes of a train edge encode the train's position on the track and the direction of its movement. 

The dynamic part of the system is specified by graph transformation rules. The rules model the movement and deletion of trains thereon. Application of the rule $\Move$ ($\Delete$) means to find an occurrence of the left-hand side in the rail net graph and to replace it with the right-hand side of the rule.
\[\begin{array}{lcl}
\Delete&=&\brule{
\;\embedtikz{
\node[waypoint,label={[below, yshift=-0.2cm]\tiny 1}] at (1,-2) (v1) {};
\node[waypoint,label={[below, yshift=-0.2cm]\tiny 2}] at (2,-2) (v2) {};
\draw[trackgray] (v1) edge (v2);\draw[trackwhite] (v1) edge (v2);
\draw[arrow] (v1) edge[bend angle=90, bend left,min distance=1em] node {\nlok{}} (v2);               	
}\;\;}
{\;\embedtikz{
\node[waypoint,label={[below, yshift=-0.2cm]\tiny 1}] at (1,-2) (v1) {};
\node[waypoint,label={[below, yshift=-0.2cm]\tiny 2}] at (2,-2) (v2) {};
\draw[trackgray] (v1) edge (v2);\draw[trackwhite] (v1) edge (v2);
\draw[draw=none] (v1) edge[white,bend angle=90, bend left,min distance=1em] node {\phantom{\nlok{}}} (v2);
}}
{\;\embedtikz{
\node[waypoint,label={[below, yshift=-0.2cm]\tiny 1}] at (1,-2) (v1) {};
\node[waypoint,label={[below, yshift=-0.2cm]\tiny 2}] at (2,-2) (v2) {};
\draw[trackgray] (v1) edge (v2);\draw[trackwhite] (v1) edge (v2);
\draw[draw=none] (v1) edge[white,bend angle=90, bend left,min distance=1em] node {\phantom{\nlok{}}} (v2);
}}\\\\

\Move&=&\brule{\;\embedtikz{
\node[waypoint,label={[below, yshift=-0.2cm]\tiny 1}] at (1,-2) (v1) {};
\node[waypoint,label={[below, yshift=-0.2cm]\tiny 2}] at (2,-2) (v2) {};
\node[waypoint,label={[below, yshift=-0.2cm]\tiny 3}] at (3,-2) (v3) {};
\draw[trackgray] (v1) edge (v2);\draw[trackwhite] (v1) edge (v2);
\draw[trackgray] (v2) edge (v3);\draw[trackwhite] (v2) edge (v3);
\draw[arrow] (v1) edge[bend angle=90, bend left,min distance=1em] node {\nlok{}} (v2);
\;}}
{\;\embedtikz{
\node[waypoint,label={[below, yshift=-0.2cm]\tiny 1}] at (1,-2) (v1) {};
\node[waypoint,label={[below, yshift=-0.2cm]\tiny 2}] at (2,-2) (v2) {};
\node[waypoint,label={[below, yshift=-0.2cm]\tiny 3}] at (3,-2) (v3) {};
\draw[trackgray] (v1) edge (v2);\draw[trackwhite] (v1) edge (v2);
\draw[trackgray] (v2) edge (v3);\draw[trackwhite] (v2) edge (v3);
\draw[draw=none] (v2) edge[white,bend angle=90, bend left,min distance=1em] node {\phantom{\nlok{}}} (v3);
}}
{\;\embedtikz{
\node[waypoint,label={[below, yshift=-0.2cm]\tiny 1}] at (1,-2) (v1) {};
\node[waypoint,label={[below, yshift=-0.2cm]\tiny 2}] at (2,-2) (v2) {};
\node[waypoint,label={[below, yshift=-0.2cm]\tiny 3}] at (3,-2) (v3) {};
\draw[trackgray] (v1) edge (v2);\draw[trackwhite] (v1) edge (v2);
\draw[trackgray] (v2) edge (v3);\draw[trackwhite] (v2) edge (v3);
\draw[arrow] (v2) edge[bend angle=90, bend left,min distance=1em] node {\nlok{}} (v3);
}\;\;}\\\\
\end{array}\]

In the following, we consider the constraint $\NoTwo$ below, meaning that there are no two trains occupying the same piece of track.

\[\NoTwo=\NE\embedtikz{
\node[waypoint,label={[below left, yshift=-0.2cm]\tiny}] at (1,-2) (v1) {};
\node[waypoint,label={[below right, yshift=-0.2cm]\tiny}] at (2,-2) (v2) {};
\draw[trackgray] (v1) edge (v2);\draw[trackwhite] (v1) edge (v2);
\draw[arrow] (v1) edge[bend angle=90, bend left,min distance=1em] node {\nlok{}} (v2);
\draw[arrow] (v1) edge[bend angle=90, bend right,min distance=1em] node {\nlok{}} (v2);
}\]

One may look for repair programs for the constraint NoTwo based on the rule sets $\{\Move\}$, $\{\Delete\}$, and $\{\Move,\Delete\}$ such that the application to any graph yields a graph satisfying the constraint $\NoTwo$. 
\ignore{It turns out that there is no $\Moved$-based repair program for $\NoTwo$: e.g, for the graph \
\[\embedtikz{
\node[waypoint,label={[below left, yshift=-0.2cm]\tiny}] at (1,-2) (v1) {};
\node[waypoint,label={[below right, yshift=-0.2cm]\tiny}] at (2,-2) (v2) {};
\draw[trackgray] (v1) edge (v2);\draw[trackwhite] (v1) edge (v2);
\draw[arrow] (v1) edge[bend angle=90, bend left,min distance=1em] node {\nlok{}} (v2);
\draw[arrow] (v1) edge[bend angle=90, bend right,min distance=1em] node {\nlok{}} (v2);
}\]}
\end{example}

The structure of the paper is as follows. 
In Section~\ref{sec:preliminaries}, we review the definitions of graphs, graph conditions, and graph programs.
In Section~\ref{sec:repair}, we introduce repair programs and show that there are repair programs for so-called proper conditions. 
In Section~\ref{sec:rb-repair}, we introduce rule-based programs, show that there are rule-based programs for transformations, and\ignore{show that there are} rule-based repair programs for proper conditions compatible with a rule set. 
In Section~\ref{sec:related}, we present some related concepts.
In Section~\ref{sec:conclusion}, we give a conclusion and mention some further work.

\section{Preliminaries}\label{sec:preliminaries}

In the following, we recall the definitions of directed, labelled graphs, graph conditions, rules and   transformations \cite{Ehrig-Ehrig-Prange-Taentzer06b}, graph programs \cite{Habel-Plump01a}, and basic transformations \cite{Habel-Pennemann09a}. 


A directed, labelled graph consists of a set of nodes and a set of edges where each edge is equipped with a source and a target node and where each node and edge is equipped with a label.
\begin{definition}[graphs \& morphisms]  A \emph{(directed, labelled) graph} (over a label alphabet $\mathcal{L}$) is a system $G=(\V_G,\E_G,\sou_G,\tar_G,\lab_{\V,G},\lab_{\E,G})$ where $\V_G$ and $\E_G$ are finite sets of \emph{nodes} (or \emph{vertices}) and \emph{edges}, $\sou_G,\tar_G\colon$ $\E_G\to \V_G$ are total functions assigning \emph{source} and \emph{target}\/ to each edge, and $\lab_{\V,G}\colon\V_G\to\mathcal{L}$, $\lab_{\E,G}\colon\E_G\to\mathcal{L}$ are total labeling functions. If $\V_G=\emptyset$, then $G$ is the \emph{empty graph}, denoted by~$\emptyset$. A~graph is \emph{unlabelled} if the label alphabet is a singleton.
Given graphs $G$ and $H$, a \emph{(graph) morphism} $g\colon G \to H$ consists of total functions $g_\V\colon\V_G\to\V_H$ and $g_\E\colon\E_G\to\E_H$ that preserve sources, targets, and labels, that is, $g_\V\circ\sou_G=\sou_H\circ g_\E$, $g_\V\circ\tar_G=\tar_H\circ g_\E$, $\lab_{\V,G}=\lab_{\V,H}\circ g_\V$, $\lab_{\E,G}=\lab_{\E,H}\circ g_\E$. The morphism $g$ is \emph{injective}\/ (\emph{surjective}\/) if $g_{\V}$ and $g_{\E}$ are injective (surjective), and an \emph{isomorphism}\/ if it is injective and surjective. In the latter case, $G$ and $H$ are \emph{isomorphic}, which is denoted by $G\cong H$. An injective morphism $g\colon G\injto H$ is an \emph{inclusion morphism} if $g_\V(v)=v$ and $g_\E(e)=e$ for all $v\in\V_G$  and all $e\in\E_G$.
\end{definition}

\longv{\begin{convention}Drawing a graph, nodes are drawn as circles with their labels (if existent) inside, and edges are drawn as arrows with their labels (if existent) placed next to them. Arbitrary graph morphisms are drawn by usual arrows $\to$, injective graph morphisms are  distinguished by $\injto$. \end{convention}}



Graph conditions are nested constructs, which can be represented as trees of morphisms equipped with quantifiers and Boolean connectives. Graph conditions and first-order graph formulas are expressively equivalent \cite{Habel-Pennemann09a}. 

\begin{definition}[graph conditions]\label{def:cond}A \emph{(graph) condition} over a graph $A$ is of the form (a) $\ctrue$ or $\exists(a,c)$ where $a\colon A \injto C$ is a proper inclusion morphism\footnote{Without loss of generality, we may assume that\ignore{the conditions are \emph{proper}, i.e.,} for all inclusion morphisms $a\colon A\injto C$  in the condition, $A$ is a proper subgraph of~$C$.} and $c$ is a condition over $C$. (b) For a condition $c$ over~$A$, $\neg c$ is a condition over~$A$.
(c) For conditions $c_i$ ($i \in I$ for some finite index set $I$\longv{\footnote{In this paper, we consider graph conditions with finite index sets.}}) over $A$, $\wedge_{i \in I} c_i$ is a condition over~$A$. 
Conditions over the empty graph~$\emptyset$ are called \emph{constraints}. In the context of rules, conditions are called \emph{application conditions}. Conditions built by (a) and (b) are called \emph{linear}. 

Any injective morphism $p\colon A\injto G$ \emph{satisfies} $\ctrue$. An injective morphism $p$ \emph{satisfies} $\PE(a,c)$ with $a\colon~A\injto C$ if there exists an injective morphism $q\colon C\injto G$ such that $q\circ a=p$ and $q$~satisfies~$c$. 
\[\tikz[node distance=1.5em,shape=rectangle,outer sep=1pt,inner sep=2pt]{
\node(P){$A$};
\node(G)[strictly below right of=P]{$G$};
\node(C)[strictly above right of=G]{$C,$};
\draw[monomorphism] (P) -- node[overlay,above](a){$a$} (C);
\draw[monomorphism] (P) -- node[overlay,below left]{$p$} (G);
\draw[altmonomorphism] (C) -- node[overlay,below right](q){$q$} (G);
\draw[draw=white] (a) -- node[overlay](tr1){=} (G);
\node(c)[outer sep=0pt,inner sep=0pt,node distance=0em,strictly right of=C]{\tikz[draw=black,fill=lightgray]{
\filldraw (0,0) -- (0.6,0.12) -- node[right,outer sep=1ex]{\footnotesize{$c$}} (0.6,0) -- (0.6,-0.12) -- (0,0);}};
\draw[draw=white] (q) -- node[overlay,sloped](tr1){$\models$} (c);
\node(Y)[node distance=0.2em,strictly right of=c]{$)$};
\node(X)[node distance=0.0em,strictly left of=P]{$\PE($};}\]

An injective morphism $p$ \emph{satisfies} $\neg c$ if $p$ does not satisfy $c$, and $p$ \emph{satisfies} $\wedge _{i \in I}c_i$ if $p$ satisfies each $c_i$ ($i \in I$). We write  $p\models c$ if $p$ satisfies the condition $c$ (over $A$). 
A condition $c$ over $A$ is \emph{satisfiable} if there is a morphism $p \colon A \injto G$ that satisfies $c$.
A graph $G$ \emph{satisfies} a constraint $c$, $G\models c$, if the morphism $p\colon\emptyset\injto G$ satisfies~$c$.
A constraint $c$ is \emph{satisfiable} if there is a graph $G$ that satisfies $c$. 

Two conditions $c$ and $c'$ over $A$ are \emph{equivalent}, denoted by $c\equiv c'$, if for all graphs $G$ and all injective morphisms $p\colon A\injto G$, $p\models c$ iff $p\models c'$. A condition $c$ \emph{implies} a condition $c'$, denoted by $c\impl c'$, if for all graphs and all injective morphisms $p\colon A\injto G$, $p\models c$ implies $p\models c'$.
\end{definition}

\begin{notation}Graph conditions may be written in a more compact form: $\PE a:=\PE(a,\ctrue)$, $\cfalse:=\neg \ctrue$ and $\PA(a,c):=\NE(a, \neg c)$, and $\NE:=\neg\PE$\ignore{, and $\NA:=\neg\PA$}. The expressions $\vee_{i \in I} c_i$ and $c\impl c'$ are defined as usual. For an inclusion morphism $a\colon A\DSinjto C$ in a condition, we just depict the codomain $C$, if the domain $A$ can be unambiguously inferred.\end{notation}

\begin{example}The expression $\neg\PE(\;\emptyset\;\injto\onenode{1}\;,\;\neg\PE(\onenode{1}\injto\twonodesedge{1}{}{}, \ctrue)\vee\neg\PE(\onenode{1}\injto\twonodesredge{1}{}{}, \ctrue))$ is a constraint according to Definition \ref{def:cond}, written in compact form as $\PA(\onenode{1}\;,\;\PE(\twonodesedge{1}{}{})\wedge\PE(\twonodesredge{1}{}{}))$ meaning that, for every node, there exists a real \footnote{An edge is said to be \emph{real}, if it is not a loop.} outgoing and a real incoming edge.
\end{example}

\begin{fact}[equivalences \cite{Pennemann09a}]\label{fac:equiv}There are the following equivalences:
\[\begin{array}{lcl@{\qquad}lcl}
\PE(x,\ctrue)&\equiv&\PE x &
\PA(x,\ctrue)&\equiv&\ctrue\\
\PE(x,\cfalse)&\equiv&\cfalse &
\PA(x,\cfalse)&\equiv&\NE x\\
\PA(x,\PE(y,\cfalse))&\equiv&\PA(x,\cfalse)\equiv\NE x &
\PE(x,\PA(y,\cfalse))&\equiv&\PE(x,\NE y)
\end{array}\]
\end{fact}

To simplify our reasoning, the repair program operates on a subset of conditions in normal form, so-called conditions with alternating quantifiers.
\begin{definition}[alternating quantifiers, proper and basic conditions]\label{def:ANF}\label{def:proper}
A linear condition of the form $\Q(a_1,\bar{\Q}(a_2,\Q(a_3,\ldots)))$ with $\Q\in\{\PA,\PE\}$, $\bar\PA=\PE$, $\bar\PE=\PA$ ending with $\ctrue$ or $\cfalse$ is a \emph{condition with alternating quantifiers (ANF)}. Such a condition in ANF is \emph{proper} if it ends with a condition $\PE(b,\ctrue)\equiv \PE b$ or it is a condition of  the form $\PE(a, \PA(b, \cfalse))\equiv\PE(a,\NE b)$ or $\PA(b, \cfalse)\equiv\NE b$. A~proper condition is \emph{basic} if it is of the form $\PE b$ or $\NE b$.
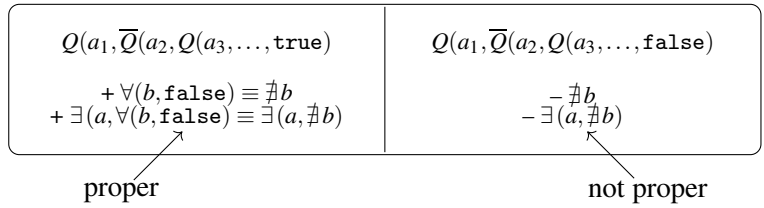
\[\tikz[node distance=1.2em,label distance=1pt]{%
 \node(0) at (0,0) {\footnotesize $Q(a_1, \bar{Q}(a_2, Q(a_3, \ldots, \ctrue)$};
 \node(1) at (5,0) {\footnotesize $Q(a_1, \bar{Q}(a_2, Q(a_3, \ldots, \cfalse)$};
 \node(2) at (0,-0.7) {\footnotesize + $\PA(b, \cfalse)\equiv\NE b$};   
 \node(2') at (0,-1) {\footnotesize + $\PE(a, \PA(b, \cfalse)\equiv \PE (a, \NE b)$};   
 \node(3) at (5,-0.7) {\footnotesize -- $\NE b$};   
 \node(3') at (5,-1) {\footnotesize -- $\PE (a, \NE b)$};   
 \node(l1) at (2.5,0.5){};
  \node(l2) at (2.5,-1.5){};
  \draw[-] (l1) to (l2); 
  \node(uld) at (-2.5,-1.5){};
  \node(org) at (7.5,0.5) {};
  \draw[rounded corners] (uld) node [above]{} rectangle (org); 
  \node (proper) at (-1,-2) {proper};
  \node (r) at (6,-2) {not proper};
  \draw[->] (proper) to (2');
  \draw[->] (r) to (3');
}\]
\end{definition}

\begin{example}The linear conditions $\PA(\onenode{1}\;,\;\PE(\twonodesedge{1}{}{},\ctrue))$ and 
$\PA(\onenode{1}\;,\;\PE(\twonodesedge{1}{}{},\cfalse))$ are conditions with alternating quantifiers. 
$\PA(\onenode{1}, \PE(\twonodesedge{1}{}{},\ctrue))$ and \scalebox{0.75}{$\PE(\onenode{1}, \PA(\twonodesedge{1}{}{}, \PE(\twonodesedgeloop{1}{2}{}{},\ctrue)))$} are  proper. Moreover, $\PA(\threenodesedgeredgeedge{1}{2}{}{}{},\cfalse)\equiv\NE(\threenodesedgeredgeedge{1}{2}{}{}{})$ is proper. The linear condition $\PA (\onenode{1},$ 
$\PE(\twonodesedgeredge{1}{}{}{},\PA(\threenodesedgeredgeedge{1}{}{}{}{},\cfalse)))\equiv\PA (\onenode{1}, \PE(\twonodesedgeredge{1}{}{}{}, \NE(\threenodesedgeredgeedge{1}{}{}{}{})))$ is non-proper.
\end{example}

By a normal form result for conditions \cite{Pennemann09a}, we obtain a normal form result for linear conditions saying that every linear condition effectively can be transformed into an equivalent condition with alternating quantifiers. 

\ignore{\red
\begin{definition}[(conjunctive) normal form \cite{Pennemann09a}]\label{def:conjnormalform}
The empty conjunction $\wedge_{j \in \emptyset} d_j \equiv \ctrue$ is in conjunctive normal form (CNF). Every condition $\wedge_{j \in J} \vee_{k \in K_j}~d_k$ is in CNF, if for every $j \in J$ and every $k \in K_j$, $d_k = \PE(a_k, c_k)$ or $d_k = \NE(a_k, c_k)$, where $a_k$ is not an isomorphism and $c_k$ is a condition in CNF.
\end{definition}
}

\begin{fact}[normal form]\label{fav:alternate} For every linear condition, there exists an equivalent condition with alternating quantifiers.\end{fact}
\begin{proof}By a conjunctive normal form result in \cite{Pennemann09a}, every condition can be effectively transformed equivalent condition in normal form. The application of the rule $\NE(a, \neg c)\equiv \PA(a, c)$ as long as possible yields an equivalent condition with alternating quantifiers.\end{proof}

By definition, proper conditions are satisfiable.
\begin{fact}[proper conditions are satisfiable]\label{fac:proper}
Every proper condition is satisfiable.
\end{fact}
\begin{proof}By Definition \ref{def:proper}, a proper condition is $\ctrue$, ends with a condition of the form $\PE(x,\ctrue)\equiv\PE x$, $\PA(x,\ctrue)\equiv\ctrue$, or is of the form $\NE b$ or $\PE(a,\NE b)$ and $b$ is not an isomorphism. Thus, it is satisfiable.\end{proof}

\begin{fact}[non-proper satisfiable conditions]\label{fac:non-proper}
There are non-proper conditions that are satisfiable.
\end{fact}

\begin{proof}The non-proper condition $\PA(\onenode{1}\onenode{2},\PE(\twonodesedge{1}{2}{},\PA(\twonodesredges{1}{2}{}{},\cfalse)))$ can be transformed into a proper one: $\PA(\onenode{1}\onenode{2},\PE(\twonodesedge{1}{2}{},\NE(\twonodesredges{1}{2}{}{})))\equiv\PA(\onenode{1}\onenode{2},\cfalse)\equiv\NE\twonodes{1}{2}$. By Fact~\ref{fac:proper}, the condition is satisfiable.
\end{proof}


Plain rules are specified by a pair of injective graph morphisms. They may be equipped with context, application conditions, and interfaces. For restricting the applicability of rules, the rules are equipped with a left application condition. {\red By extending the rules with a context, it is possible to require an application condition over an extended left-hand side (see Example \ref{ex:rb-repair})}. By the interfaces, it becomes possible to hand over information between the transformation steps. 

\begin{definition}[rules and  transformations]\label{def:rule} 
A~\emph{plain rule} $p=\brule{L}{K}{R}$ consists of two inclusion morphisms $K\injto L$ and $K\injto R$.  
The rule $p$ \emph{equipped with context} $K\injto K'$ is the rule $\brule{L'}{K'}{R'}$ where $L'$ and $R'$ are the pushout objects in the diagrams (1) and (2) below.
\[\begin{tikzpicture}[node distance=2em,shape=rectangle,outer sep=1pt,inner sep=2pt,label distance=-1.25em]
\node(L){$L$};
\node(K)[strictly right of=L]{$K$};
\node(R)[strictly right of=K]{$R$};
\node(D)[strictly below of=K]{$K'$};
\node(G)[strictly below of=L]{$L'$};
\node(H)[strictly below of=R]{$R'$};
\draw[altmonomorphism] (K) -- node[overlay,above]{\small } (L);
\draw[monomorphism] (K) -- node[overlay,above]{\small } (R);
\draw[altmonomorphism,dashed] (D) -- (G);
\draw[monomorphism,dashed] (D) -- (H);
\draw[monomorphism,dashed] (L) -- node[overlay,left](g){\small } (G);
\draw[monomorphism] (K) -- node[overlay,left]{\small } (D);
\draw[monomorphism,dashed] (R) -- node[overlay,right]{\small }(H);
\draw[draw=none] (L) -- node[overlay]{\small (1)} (D);
\draw[draw=none] (R) -- node[overlay]{\small (2)} (D);
 \end{tikzpicture}\]
A \emph{rule} $\prule=\tuple{x,p,\ac,y}$ \emph{with interfaces} $X$ and $Y$ consists of a plain rule $p=\brule{L}{K}{R}$ with left application condition $\ac$ and two injective morphisms $x\colon X\injto L$, $y\colon Y\injto R$, called the \emph{(left and right) interface morphisms}. If both interfaces are empty, i.e., the domains of the interface morphisms are empty, we write $\prule=\tuple{p,\ac}$. If additionally $\ac=\ctrue$, we write $\prule=\tuple{p}$ or short $p$. 
A \emph{direct transformation}\/ $G\dder_{\prule,g,h,i} H$ or short $G\dder_\prule H$ from $G$ to $H$ applying $\prule$ at $g\colon X\injto G$ consists of the following steps: 
\begin{enumerate}
\item[(1)] Select a match $g'\colon L\injto G$ such that $g=g'\circ x$ and $g'\models \ac$.
\item[(2)] Apply the plain rule\footnote{The application of a plain rule is as in the double-pushout approach \cite{Ehrig-Ehrig-Prange-Taentzer06b}.} $p$ at $g'$ (possibly) yielding a comatch $h'\colon R\injto H$.
\item[(3)]Unselect $h\colon Y\injto H$, i.e., define $h=h'\circ y$.
\end{enumerate}

\begin{figure}[h]
\[\scalebox{1.2}{
\begin{tikzpicture}[node distance=2em,shape=rectangle,outer sep=1pt,inner sep=2pt,label distance=-1.25em]
\node(X){$X$};
\node(L)[strictly right of=X]{$L$};
\node(K)[strictly right of=L]{$K$};
\node(R)[strictly right of=K]{$R$};
\node(Y)[strictly right of=R]{$Y$};
\node(D)[strictly below of=K]{$D$};
\node(G)[strictly left of=D]{$G$};
\node(H)[strictly right of=D]{$H$};
\draw[monomorphism] (X) -- node[overlay,above]{\footnotesize $x$} (L);
\draw[altmonomorphism] (K) -- node[overlay,above]{\footnotesize $l$} (L);
\draw[monomorphism] (K) -- node[overlay,above]{\footnotesize $r$} (R);
\draw[altmonomorphism] (Y) -- node[overlay,above]{\footnotesize $y$} (R);
\draw[altmonomorphism] (D) -- node[overlay,above]{\footnotesize $l^{*}$}(G);
\draw[monomorphism] (D) -- node[overlay,above]{\footnotesize $r^{*}$}(H);
\draw[monomorphism] (X) -- node[overlay,below left](g){\footnotesize $g$} (G);
\draw[monomorphism] (L) -- node[overlay,left]{\footnotesize $g'$} (G);
\draw[monomorphism] (K) -- node[overlay,left]{\footnotesize } (D);
\draw[monomorphism] (R) -- node[overlay,right]{\footnotesize $h'$}(H);
\draw[monomorphism] (Y) -- node[overlay,below right](h){\footnotesize $h$}(H);
\draw[morphism,dashed] (X) edge[bend left=30] node[above]{\footnotesize $i$} (Y);
\draw[morphism,dashed] (G) edge[bend right=30] node[below]{\footnotesize $\tr$} (H);
\draw[draw=none] (L) -- node[overlay]{\footnotesize (1)} (D);
\draw[draw=none] (R) -- node[overlay]{\footnotesize (2)} (D);
\draw[draw=white] (g) -- node[overlay]{=} (L);
\draw[draw=white] (h) -- node[overlay]{=} (R);
\node(c)[outer sep=0pt,inner sep=0pt,node distance=0em,strictly above of=L]{\tikz[draw=black,fill=lightgray]{
\filldraw (0,0) -- (-0.12,0.5) -- node[above,outer sep=1ex]{\footnotesize{$\ac$}} (0.12,0.5) -- (0,0);}};
\end{tikzpicture}}\]
\caption{\label{fig:DPO}A direct transformation}
\end{figure}

A triple $\tuple{g,h,i}$ with partial\footnote{A \emph{partial} morphism $i\colon X\injpar Y$ is an injective morphism $X'\injto Y$ such that $X'\subseteq X$.}  morphism $i=y^{-1} \circ r \circ l^{-1}\circ x$ (called \emph{interface relation}) is in the semantics of $\prule$, denoted by $\psem{\prule}$, if there is an injective morphism $g'\colon L\injto G$ such that $g=g'\circ x$ and $g'\models\ac$, $G\dder_{p,g',h'} H$, and $h=h'\circ y$. We write $G\dder_{\prule,g,h,i} H$ or short $G\dder_{\prule} H$. 
Given graphs $G$, $H$\/  and a finite set $\R$\/ of rules, $G$ \emph{derives} $H$\/ by $\R$\/ if $G \cong H$ or there is a sequence of direct transformations $G=G_0\dder_{\prule_1,g_1,h_1}G_1\dder_{\prule_2,h_1,h_2}\ldots\dder_{\prule_n,h_{n-1},h_n}G_n=H$ with $\prule_1,\dots,\prule_n\in \R$. In this case, we write $G\der_{\R} H$\/ or just $G \der H$.
\end{definition}

\begin{notation}
If both interfaces of $\prule=\tuple{x,p,\ac,y}$ are empty, we write $\prule=\tuple{p,\ac}$. 
If additionally $\ac=\ctrue$, we write $\prule=\tuple{p}$ or short $p$.
A plain rule $p=\brule{L}{K}{R}$ sometimes is denoted by $L\dder R$ where indexes in $L$ and $R$ refer to the corresponding nodes. 
Moreover, $\select(x,\ac)$ and $\unselect(x)$ denote the rules  $\tuple{x,\id,\ac}$ (selection of elements) and 
$\tuple{\id,\y}$ (unselection of selected elements), respectively,  
where $\id$ denotes the identical rule $\brule{L}{L}{L}$. Additionally, $\select(x)$ abbreviates $\select(x,\ctrue)$.
\ignore{We use the following abbreviations:
\[\begin{array}{lcll}
\select(x,\ac)&::=&\tuple{x,\id,\ac}&\mbox{(selection of elements)}\\
\unselect(x)&::=&\tuple{\id,\y}&\mbox{(unselection of selected elements)}\\
\end{array}\]
where $\id=\brule{L}{L}{L}$ demotes the \emph{identical} plain rule with left-hand side $L$. Additionally, $\select(x)$ abbreviates $\select(x,\ctrue)$.}
\end{notation}

\begin{example}\label{ex:interface}
Consider the rule $\prule=\tuple{x,p,y}$ with the plain rule $p=\brule{\onenode{1}}{\onenode{1}}{\twonodesedge{1}{2}{}}$, and the interface morphisms $x\colon\onenode{1}\injto \onenode{1}$, $y\colon\onenode{1}\injto\twonodesedge{1}{2}{}$ (see Figure \ref{fig:DPOex}). 
\begin{figure}[h!]
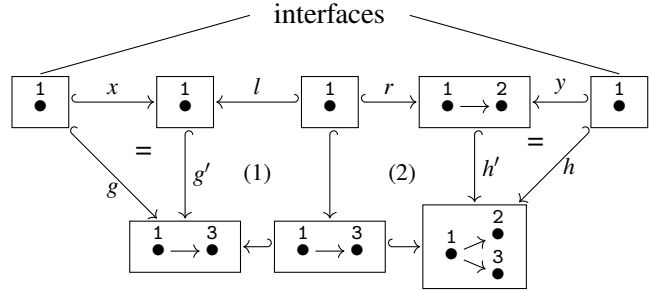

\[
\tikz[node distance=5em,shape=rectangle,outer sep=1pt,inner sep=2pt,label distance=-1.25em]{
\node(G)[draw=black]{$\twonodesedge{1}{3}{}$};
\node(D2)[draw=black,right of=G]{$\twonodesedge{1}{3}{}$};
\node(H2)[draw=black,right of=D2]{$\interfacestar{1}{2}{3}$};
\node(R2)[draw=black,above of=H2]{$\twonodesedge{1}{2}{}$};
\node(K2)[draw=black,above of=D2]{$\onenode{1}$};
\node(L2)[draw=black,left of=K2]{$\onenode{1}$};
\node(X1)[left of=L2,draw]{$\onenode{1}$};
\node(Y1)[right of=R2,draw]{$\onenode{1}$};
\node(i)[above of=K2,node distance=3em]{\begin{tabular}{c}interfaces \end{tabular}};
\draw[-] (i.west) edge node[above]{} (X1.north);
\draw[-] (i.east) edge node[above]{} (Y1.north);
\draw[altmonomorphism] (K2) -- node[above]{\footnotesize $l$} (L2);
\draw[monomorphism] (K2) -- node[above]{\footnotesize $r$} (R2);
\draw[altmonomorphism] (D2) -- node[below]{\footnotesize} (G);
\draw[monomorphism] (D2) -- node[below]{\footnotesize } (H2);
\draw[draw=none] (L2) -- node[overlay]{\footnotesize (1)} (D2);
\draw[draw=none] (R2) -- node[overlay]{\footnotesize (2)} (D2);
\draw[monomorphism] (L2) -- node[right]{\footnotesize$g'$} (G);
\draw[monomorphism] (K2) -- (D2);
\draw[monomorphism] (R2) -- node[right]{\footnotesize $h'$} (H2);
\draw[monomorphism] (X1) -- node[above](x){\footnotesize $x$} (L2);
\draw[altmonomorphism] (Y1) -- node[above](y){\footnotesize $y$} (R2);
\draw[monomorphism] (X1) -- node[below](g){\footnotesize $g$}  (G);
\draw[altmonomorphism] (Y1) --  node[right](h){\footnotesize $h$}  (H2);
\draw[draw=white] (g) -- node[overlay]{=} (L2);
\draw[draw=white] (h) -- node[overlay]{=} (R2);
}\]
\caption{\label{fig:DPOex}A direct example transformation}
\end{figure}

Each injective morphism~$g\colon\onenode{}~\injto~\twonodesedge{}{}{}$ fixes a node in the host graph. 
In general, the morphism~$g$ restricts the allowed matches $g'$ from the left-hand-side into the host graph by $g'=g\circ x$. The plain rule is applied at $g'$ according the double-pushout approach yielding the comatch $h'\colon\twonodesedge{1}{2}{}\injto \interfacestar{1}{2}{3}$. Defining $h=h'\circ y$, we fix the node 1 for the next rule application. It says that at this position (and no other) the rule shall be applied.
\end{example}


With every transformation $t\colon G\der H$, a partial track morphism can be associated that ``follows'' the items of $G$ through the transformation: this morphism is undefined for all items in $G$ that are removed by~$t$, and maps all other items to the corresponding items in $H$.
\begin{definition}[track morphism \cite{Plump05a}]The \emph{track morphism} $\tr_{G \dder H}$ from $G$ to $H$ is the partial morphism defined by $\tr_{G \dder H}(x)=r^*{(l^*}^{-1}(x))$ if $x\in D$ and \emph{undefined} otherwise, where the morphisms $l^*\colon D\injto G$ and $r^*\colon D\injto H$ are the induced morphisms of $l\colon K\injto L$ and $r\colon K\injto R$, respectively (see Figure~\ref{fig:DPO}). 
Given a transformation $G \der H$, $\tr_{G \der H}$ is defined by induction on the length of the transformation:  $\tr_{G \der H}=\iso$ for an isomorphism $\iso\colon G\to H$ and $\tr_{G \der H}=\tr_{G'\dder H}\circ\tr_{G\der G'}$ for $G \Rightarrow^{+} H = G\der G'\dder~H$.
\end{definition}
\begin{example}\label{ex:interface2}For the direct transformation $t\colon G\dder_\prule H$ in Example~\ref{ex:interface}, the track morphism is $\tr_t\colon\twonodesedge{1}{3}{}\injto\interfacestar{1}{2}{3}$. For the direct inverse transformation $t'$, $\tr_{t'}\colon\interfacestar{1}{2}{3}\rightharpoonup \twonodesedge{1}{3}{}$ is partial.\end{example}
\ignore{\begin{example}For the transformation below, the track morphism is drawn in dotted lines. It is undefined for the left node and the edge, which are deleted by the transformation, and maps the preserved items, i.e., the two right nodes, and the edge in $G$ to the corresponding node and edge in $H$.
\[\scalebox{1.3}{
\tikz[node distance=1.5em,label distance=1pt]{
\node(0){$\circ$};
\node(1)[strictly below left of=0]{$\circ$};
\node(2)[strictly below right of=0]{$\circ$};
\draw[->](0) to node(a){} (1);
\draw[->](0) to node(b){} (2);
\node(0')[strictly right of=0,node distance=8em]{$\circ$};
\node(1')[strictly below left of=0']{\white $\circ$};
\node(2')[strictly below right of=0']{$\circ$};
\draw[->](0') to node(b'){} (2');
\draw[->,dotted] (0) to (0');
\draw[->,dotted] (2) to (2');
\draw[->,dotted,] (b) to (b');
\node(ulg)[below left of=1,node distance=1.5em] {};
\node(hg)[right of =2,node distance=1em] {};
\node(org)[above of =hg,node distance=2.6em] {};
\draw[gray] (org) rectangle (ulg);    
\node(G)[above of=1,node distance=2em]{$G$}; 
\node(ulh)[below left of=1',node distance=1.5em] {};
\node(hh)[right of =2',node distance=1em] {};
\node(orh)[above of =hh,node distance=2.5em] {};
\draw[gray] (orh) rectangle (ulh);    
\node(H)[above of=1',node distance=2em]{$H$}; 
}}\]
\end{example}}


Graph programs are made of sets of rules with interface, {non-deterministic choice, sequential composition,  as-long-as possible iteration, and the try-statement. 
\begin{definition}[graph programs]\label{def:prog}\label{terminating} The set of \emph{(graph) programs with interface $X$}, $\Prog(X)$, is defined inductively: Consider
\[\begin{tabular}{lll}
(1)& Every rule $\prule$ with interface $X$ (and $Y$) is in $\Prog(X)$.\\
(2)& If $P,Q\in\Prog(X)$, then $\{P,Q\}$ is in $\Prog(X)$ &(nondeterministic choice).\\
(3)& If $P\in\Prog(X)$ and $Q\in\Prog(Y)$, then $\tuple{P;Q}\in\Prog(X)$ & (sequential composition).\\
(4)& If $P\in\Prog(X)$, then $P\downarrow$, and $\try P$ are in $\Prog(X)$ & (iteration \& try).\\
\end{tabular}\]
\newpage 
The \emph{semantics} of a program $P$ with interface $X$, denoted by $\psem{P}$, is a set of triples  such that, for all $\tuple{g,h,i}\in\psem{P}$, $X=\dom(g)=\dom(i)$\footnote{For a partial morphism $i$, $\dom(i)$ and $\ran(i)$ denote the domain and codomain of~ $i$, respectively.} and $\dom(h)=\ran(i)$, and is defined as follows:
\[\begin{array}{llcl}
(1) &\psem{\prule}&&\mbox{as in Definition \ref{def:rule}}\\
(2)&\psem{\{P,Q\}} &=& \psem{P}\cup\psem{Q}\\
(3) &\psem{\tuple{P;Q}}&=&\{\tuple{g_1,h_2,i_2{\circ}i_1}\mid \tuple{g_1,h_1,i_1}{\in}\psem{P}, \tuple{g_2,h_2,i_2}{\in}\psem{Q}\mbox{ and }h_1=g_2\}\\
(4) &\psem{\aslong{P}}& = &\{\tuple{g,h,\id}\in P^*\mid\nexists h'.\tuple{h,h',\id}\in\psem{\Fix(P)}\} \\
&\psem{\try P}& =& \{\tuple{g,h,i} \mid \tuple{g,h,i} \in \psem{P}\} \cup \{\tuple{g,g,\id} \mid \NE h.\tuple{g,h,i} \in \psem{P}\}
\end{array}\]
where $P^*=\bigcup_{j=0}^\infty P^j$ with $P^0 = \Skip$, $P^{j} = \tuple{\Fix(P); P^{j-1}}$ for $j>0$ 
and $\psem{\Fix(P)}=\{\tuple{g,h\circ i,\id}\mid \tuple{g,h,i}\in\psem{P}\}$. 
Two programs $P,P'$ are \emph{equivalent}, denoted $P\equiv P'$,  if $\psem{P}=\psem{P'}$. 
A program $P$ is \emph{terminating} if the relation $\to$ is terminating. \ignore{i.e., the is no infinite chain.} 
\end{definition}

The statement $\Skip$ is the identity element $\select(\id, \ctrue)$ of sequential composition.

\begin{example}\label{ex:interface3}Consider a slightly modified example as in Example \ref{ex:interface}. 
For restricting the applicability of the plain rule $\AddEdge=\brule{\onenode{1}}{\onenode{1}}{\twonodesedge{1}{2}{}}$ to a fixed node, the rule is equipped with a right interface $y_1\colon\onenode{1}\injto\twonodesedge{1}{2}{}$ yielding the rule $\AddEdge_1=\tuple{\AddEdge,y_1}$ as well as with a left interface $x_2\colon\onenode{1}\injto\onenode{1}$ yielding the rule $\AddEdge_2=\tuple{\AddEdge,x_2}$. By Definitions \ref{def:rule} and  \ref{def:prog}, $h'_1\circ y_1=h_1=g_2=g'_2\circ x_2$, i.e., the middle diagram commutes (see Figure \ref{fig:sequence}). 

\begin{figure}[h]
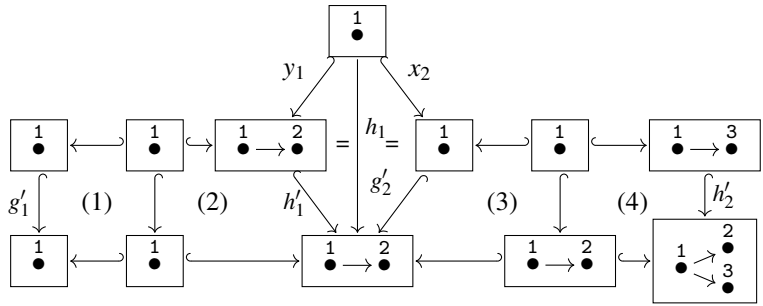

\[\tikz[node distance=4em,shape=rectangle,outer sep=1pt,inner sep=2pt,label distance=-1.25em]{
\node(space){};
\node(G)[draw=black,node distance=4em,strictly below of=space]{$\twonodesedge{1}{2}{}$};
\node(D1)[draw=black,node distance=7em,left of=G]{$\onenode{1}$};
\node(H1)[draw=black,left of=D1]{$\onenode{1}$};
\node(R1)[draw=black,above of=H1]{$\onenode{1}$};
\node(K1)[draw=black,right of=R1]{$\onenode{1}$};
\node(L1)[draw=black,right of=K1]{$\twonodesedge{1}{2}{}$};
\node(L1aux)[node distance=2.5em,right of=L1]{\footnotesize =};
\node(D2)[draw=black,node distance=7em,right of=G]{$\twonodesedge{1}{2}{}$};
\node(H2)[draw=black,node distance=5em,right of=D2]{$\interfacestar{1}{2}{3}$};
\node(R2)[draw=black,above of=H2]{$\twonodesedge{1}{3}{}$};
\node(K2)[draw=black,above of=D2]{$\onenode{1}$};
\node(L2)[draw=black,left of=K2]{$\onenode{1}$};
\node(L2aux)[node distance=1.8em,left of=L2]{\footnotesize =};
\node(K)[draw=black,strictly above of=G,node distance=6em]{$\onenode{1}$};
\draw[altmonomorphism] (K1) -- (R1);\draw[monomorphism] (K1) -- (L1);
\draw[altmonomorphism] (K2) -- (L2);\draw[monomorphism] (K2) -- (R2);
\draw[altmonomorphism] (D1) --  node[below]{\footnotesize} (H1);
\draw[monomorphism] (D1) -- node[below]{\footnotesize } (G);
\draw[altmonomorphism] (D2) -- node[below]{\footnotesize} (G);
\draw[monomorphism] (D2) -- node[below]{\footnotesize } (H2);
\draw[monomorphism] (L1) -- node[left]{\footnotesize{$h'_1$}} (G);
\draw[monomorphism] (K1) -- (D1);
\draw[monomorphism] (R1) -- node[left]{\footnotesize{$g'_1$}} (H1);
\draw[monomorphism] (L2) -- node[above left]{\footnotesize{$g'_2$}} (G);
\draw[monomorphism] (K2) -- (D2);
\draw[monomorphism] (R2) -- node[right]{\footnotesize{$h'_2$}} (H2);
\draw[altmonomorphism] (K) -- node[above left]{\small $y_1$} node[below right]{\footnotesize \ignore{$=$}} (L1);
\draw[monomorphism] (K) --  node[above right]{\small $x_2$} node[below left]{\footnotesize \ignore{$=$}}  (L2);
\draw[draw=none] (R1) -- node[overlay]{\small (1)} (D1);
\draw[morphism] (K) --  node[above right]{\footnotesize $h_1$}  (G);
\draw[draw=none] (L1) -- node[overlay]{\small (2)} (D1);
\draw[draw=none] (L2) -- node[overlay]{\small (3)} (D2);
\draw[draw=none] (R2) -- node[overlay]{\small (4)} (D2);
}\]
\caption{\label{fig:sequence}A sequence of direct transformations}
\end{figure}
\end{example}


\ignore{In the double-pushout approach, the \emph{dangling condition} for a rule $\prule=\bshortrule{L}{K}{R}$ and an injective morphism $g\colon L\injto G$ requires: ``No edge in $G{-}g(L)$ is incident to a node in $g(L{-}K)$''. For this condition, there is a program such that after application the dangling condition is satisfied.
\begin{fact}[guarantee of rule application \cite{Habel-Sandmann18a}]
For every\ignore{node-deleting\footnote{A rule $p = \tuple{L\injlto K\injto R}$ is \emph{node-preserving} if $|V_L| = |V_K|$. It is \emph{node-deleting} if it is not node-preserving, i.e. $|V_L| > |V_K|$.}} rule $\prule$ with left-hand side $L$ and every match $g \colon L \injto G$, there is an edge-deleting program $P_{\acdang}(\prule)$,
such that $\PA G \dder_{P_\acdang(\prule)} H$, with partial track morphism $\tr\colon G \injto G'$,
there is an application $G \dder_{P_\acdang(\prule)} G' \dder_{\prule, \tr \circ g} H$ of the rule $\prule$.
\end{fact}}


The construction $\Shift$ ``shifts'' existential conditions over morphisms into a disjunction of existential application conditions. 

 \begin{lemma}[$\Shift$ \cite{Habel-Pennemann09a}]\label{lem:shift}\label{lem:left}\label{lem:pres} 
There is a construction $\Shift$, such that the following holds. 
Let $d$ be condition  over $A$ and $b\colon A\injto R,n\colon R\injto H$. Then $n\circ b\models d\iff n \models \Shift(b,d)$.
\end{lemma}

\begin{construction}\label{const:Shift} \label{const:Left}\label{const:pres}
For rules $\prule$ with plain rule $p=\brule{L}{K}{R}$, the construction is as follows.

\[\begin{tabular}{l} 
\tikz[shape=rectangle,node distance=2em,shape=circle,outer sep=0pt,inner sep=1pt]{
\node(P){$A$};
\node(C)[strictly below of=P]{$C$};
\node(space)[node distance=1.5em,strictly below of=C]{};
\node(P')[strictly right of=P]{$R$};
\node(C')[strictly right of=C]{$R'$};
\draw[monomorphism] (P) -- node[overlay,left]{\small $a$} (C);
\draw[monomorphism,dashed] (P') -- node[overlay,right,inner sep=0pt]{\small $a'$} (C');
\draw[draw=none] (P) -- node[overlay]{(1)} (C');
\draw[monomorphism] (P) -- node[overlay,above]{\small $b$} (P');
\draw[monomorphism,dashed] (C) -- node[overlay,below]{\small $b'$} (C');
\node(c)[outer sep=0pt,inner sep=0pt,node distance=0em,strictly below of=C]{
\tikz[baseline,draw=black,fill=lightgray]{\filldraw (0,0) -- node[below,pos=0.6,overlay,outer sep=1ex]
{\small $d$} (0.12,-0.3) -- (-0.12,-0.3) -- (0,0);}};
\node(c')[outer sep=0pt,inner sep=0pt,node distance=0em,strictly below of=C']{
\tikz[baseline,draw=black,fill=lightgray]{\filldraw (0,0) -- node[below,pos=0.6,overlay,outer sep=1ex]
{\small } (0.12,-0.3) -- (-0.12,-0.3) -- (0,0);}};}
\end{tabular}
\hspace{0.3cm}
\hfill
\begin{tabular}{l} 
$\Shift(b,\ctrue):=\ctrue$.\\
$\Shift(b,\PE(a,d)):=\bigvee_{(a',b')\in\F}\PE(a',\Shift(b',d))$ where \\
$\F=\{(a',b')\mid\mbox{$b'\circ a=a'\circ b$, $a',b'$ inj, $(a',b')$ jointly surjective\footnotemark}\}$\\
$\Shift(b,\neg d):=\neg\Shift(b,d)$, $\Shift(b,\wedge_{i\in I}d_i):=\wedge_{i\in I}\Shift(b,d_i)$.
\end{tabular}\]
\footnotetext{A pair $(a',b')$ is \emph{jointly surjective} if for each $x\in R'$ there is a preimage $y\in R$ with $a'(y)=x$ or $z\in C$ with $b'(z)=x$.}

\end{construction}
\footnotetext{For a rule $p=\brule{L}{K}{R}$, $p^{-1}=\brule{R}{K}{L}$ denotes the \emph{inverse} rule. For $L'\dder_p R'$ with intermediate graph $K'$, $\brule{L'}{K'}{R'}$ is the \emph{derived} rule.}

\begin{example}[$\Shift$] The application of Shift to  the injective morphism $b\colon\emptyset\injto\onenode{1}$ and the condition $d=\PE(\emptyset\injto \onenodeloop{2}{})$ yields the condition $\Shift(b,d)=\PE(\onenode{1}\injto\onenode{1}\onenodeloop{2})\vee\PE(\onenode{1}\injto\onenodeloop{1=2})$. The application of $\Shift$ to $b$ and  of the condition $d'=\PE(\emptyset\injto \onenodeloop{2},\NE\onenodeloops{2})$ over $b$ yields the  condition $\Shift(b,d')=\PE(\onenode{1}\injto\onenode{1}\onenodeloop{2},\NE\onenode{1}\onenodeloops{2})\vee\PE(\onenode{1}\injto\onenodeloop{1=2},\NE\onenodeloops{1=2})$.
\end{example}

\section{Graph repair}\label{sec:repair}

In this section, we define repair programs and look for\ignore{maximally preserving, terminating} repair programs for \ignore{(proper) }graph conditions. 

A~repair program for a constraint is a program such that, for every application to a graph, the resulting graph satisfies the constraint. More generally, a repair program for a condition over a graph $A$ is a program~$P$ with interface $A$ such that for  every triple $\tuple{g,h,i}$ in the semantics of $P$, the composition of the interface relation $i$ and the comatch $h$ satisfies the condition.

\begin{definition}[repair programs]\label{def:repair} 
A program $P$  is a \emph{repair program} for a constraint~$d$ if, for all transformations $G\dder_P H$, $H\models d$. An $A$-preserving program \footnote{A program is \emph{$A$-preserving}  if the dependency relation $i$ of the program is total. If, additionally, the codomain of $i$ is $A$, the program is a \emph{program with interfaces} $A$ or short \emph{$A$-program.} For a rule set with interfaces $A$, we speak of \emph{$A$-set}.} $P$ is a \emph{repair program} for a condition $d$ over $A$, if, for all triples $\tuple{g,h,i}\in\psem{P}$, $h\circ i\models\ d$. 
\end{definition}
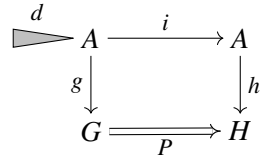
\vspace{-0.5em}
\[\tikz[node distance=2em,shape=rectangle,outer sep=1pt,inner sep=2pt]{
\node(A1){$A$};
\node(A2)[node distance=4em,strictly right of=A1]{$A$};
\node(G)[strictly below of=A1]{$G$};
\node(H)[strictly below of=A2]{$H$};
\draw[morphism] (A1) -- node[above]{\footnotesize $i$} (A2);
\draw[derivation] (G) -- node[below]{\footnotesize $P$} (H);
\draw[morphism,] (A1) -- node[left]{\footnotesize $g$} (G);
\draw[morphism] (A2) -- node[right]{\footnotesize $h$} (H);
\node(acL)[outer sep=0pt,inner sep=0pt,node distance=0em,strictly left of=A1]
{\tikz[baseline,draw=black,fill=lightgray]{\filldraw (0,0) -- node[above,pos=0.6,overlay,outer sep=1ex](acL2){\footnotesize $d$} (-0.8,0.12) -- (-0.8,-0.12) -- (0,0);}};
}\]
\begin{example}\label{ex:exists}
For the condition $c=\PE(\onenode{1}\injto\onenodeloop{1}{})$, the $\onenode{}$-preserving program $P_{c} = \try \R$ is a \ignore{\red terminating} repair program for $c$, where $\R = \tuple{x,\onenode{1}\dder\onenodeloop{1}{},\NE\onenodeloop{1}{},y}$, with the interface morphisms $x \colon \onenode{1} \injto \onenode{1}, y \colon \onenodeloop{1}{} \injlto \onenode{1}$.
For the constraint $d=\PA(\onenode{},\PE\onenodeloop{}{})$, meaning that every node has a loop, the program 
$\tuple{\select(\emptyset \injto \onenode{1});P_{c};\unselect(\onenode{1} \injlto \emptyset)}\downarrow$ is a \ignore{\red terminating} repair program for $d$. 
\end{example}

\begin{remark}
A program for a condition is \emph{destructive}, if it deletes the input graph and creates a
graph satisfying the condition from the empty graph.
In general, destructive programs are no repair programs for $d$ over $A\not=\emptyset$, because it is not $A$-preserving. 
\end{remark}

The most significant point are the repair programs for the basic conditions $\PE a$ and $\NE a$. 
Whenever we have repairing sets, we obtain a repair program for proper conditions.

\begin{definition}[repairing sets]\label{def:repair}Let $a\colon A\injto C$ with $A\subset C$. 
An $A$-set $\R_a$ is \emph{repairing} for $\PE a$ if $\try \R_a$ is a repair program for $\PE a$. An $A$-set $\S_a$ is \emph{repairing} for $\NE a$ if $\Rdown{\S'_a}$ (see Definition \ref{def:dangling}) is a repair program for~$\NE a$.\end{definition}

\begin{example}
For the condition $c$ from Example \ref{ex:exists}, the repairing set is $\R$.
\end{example}

\begin{definition}[The dangling-edges operator]\label{def:dangling} 
For node-deleting rules $\prule$, the dangling condition\footnote{The \emph{dangling condition} for a rule $\prule=\brule{L}{K}{R}$ and an injective morphism $g\colon L\injto G$ requires: ``No edge in $G{-}g(L)$ is incident to a node in $g(L{-}K)$''.} may be not satisfied. In this case, we consider the program $\prule'$ that fixes a match for the rule, deletes the dangling edges, and afterwards applies the rule at the match. The program corresponds with the SPO-way of rewriting \cite{Loewe93a}. The proceeding can be extended to sets of rules: For a rule set $\S$, $\S'=\{\prule'\mid\prule\in\S\}$. \\
\end{definition}
\vspace{-1em}
In the following, we show that, for basic conditions $\PE(A\injto C)$ and $\NE(A\injto C)$ over $A$, there are repairing $A$-sets $\R_a$ and $\S_a$, respectively.
The rules in $\R_a$ are increasing and of the form $B\dder C$ where $A\subseteq B\subset C$ and an application condition requiring that no larger subgraph $B'$ of $C$ occurs and the shifted condition $\NE a$ is satisfied. By the application condition, each rule can only be applied iff the condition is not satisfied and no other rule whose left-hand side includes $B$ and is larger can be applied. 
The rules in  $\S_a$ are decreasing and of the form $C\dder B$ where $A\subseteq B\subset C$ such that, if the number of edges in $C$ is larger than the one in $A$, they delete one edge and no node, and delete a node, otherwise. By $B\subset C$, both rule sets do not contain identical rules. 
The rule set $\R_a$ can be used, e.g., for the repair program of the condition $\PA(x,\PE a)$, the rule set $\S_a$ for the condition $\PE (x, \NE a)$ (see Construction~\ref{const:proper}).

\begin{lemma}[basic repair]\label{lem:basic}
For basic conditions over $A$, there are repairing sets with interfaces $A$.
\end{lemma}

There are several repairing sets for a basic condition: We present two examples of repairing sets. The first one  is quite intuitive, but, in general does not lead to a terminating and maximally preserving repair program. The second one is more complicated, but yields a terminating and maximally preserving repair program. 

\begin{construction}\label{const:basic} For  $d=\PE a$ ($\NE a$) with $a\colon A \injto C$, $A\subset C$, the sets $\R_a$ and $\S_a$ are constructed as follows.
\begin{enumerate}
\item[(1)] $\R_a=\{\tuple{\id_A,A\dder C,a}\}$ and $\S_a=\{\tuple{a,C\dder A,\id_A}\}$
\item[(2)] $\R_a=\{\tuple{b,B\dder C,\ac\wedge\ac_B,a}\mid A\injto^b B\subset C\}$ and  
$\S_a=\{\tuple{a,C\dder B,b}\mid A\injto^b B\subset C\mbox{ and (*)}\}$\\
where $\ac=\Shift(A\injto B,\NE a)$, $\ac_B=\bigwedge_{B'}\NE B'$, $\bigwedge_{B'}$ ranges over $B'$ with $B\subset B'\subseteq C$, and\\(*) $\pif \E_C\supset \E_B \pthen |\V_C|=|\V_B|,|\E_C|=|\E_B|+1 \pelse |\V_C|=|\V_B|+1$.
\end{enumerate}
\end{construction}

\begin{proof} (1) Let $\tuple{g,h,i}\in\psem{\try\R_a}$.  If $g\models\NE a$, then the rule $\tuple{\id_A,A\dder C,a}$ in $\R_a$ is applicable and $h\circ i\models\PE a$. If $g\models\PE a$, then, by the semantics of $\try$, $g=h\circ i\models\PE a$. Thus, $\try\R_a$ is a repair program for $\PE a$. Let $\tuple{g,h,i}\in\psem{\Rdown{\S'_a}}$. By the semantics of~$\downarrow$, $\S'_a$ is not applicable to the domain of $h\circ i$, i.e., $h\circ i\models\NE a$. Thus, $\Rdown{\S'_a}$ is a repair program for $\NE a$. For Construction (2), see the  proof \cite[Thm 1]{Habel-Sandmann18a}.
\end{proof}

\begin{example}\label{ex:Ra}
1. Consider the condition $d = \PE a$, with $a\colon \onenode{1}\injto\twonodesedge{1}{}{}$.
By Construction (1), the rule $\prule=\tuple{x,p,y}$ with the plain rule $p=\brule{\onenode{1}}{\onenode{1}}{\twonodesedge{1}{2}{}}$ and the interface morphisms $x\colon\onenode{1}\injto \onenode{1}$, $y\colon\onenode{1}\injto\twonodesedge{1}{}{}$ constitutes the repairing set for $\PE a$. 
By Construction (2), we obtain a repairing set $\R_a$ for $\PE a$.
\[\R_a=\left\{\begin{array}{lcl}
\prule_1=\tuple{x_1, \onenode{1}&\dder&\twonodesedge{1}{}{},\NE\twonodes{1}{}, y_1}\\
\prule_2=\tuple{x_2, \twonodes{1}{2}&\dder&\twonodesedge{1}{2}{},
\NE\twonodesedge{1}{2}{}\wedge\NE\twonodesedge{1}{}{}\onenode{2}, y_2}\\
\end{array}\right.\]
where the interface morphisms $x_i, y_i$ can be unambiguously inferred. 
The first rule requires a node and attaches a node and a real outgoing edge, provided that there do not exist two nodes. The second rule requires two nodes and attaches a real outgoing edge provided there is no real outgoing edge from the image of node 1 to the image of node 2, and no real outgoing edge at the image of node~1.
The rule set $\R_a$ can be used, e.g., for a repair program $\try \R_a$ for the condition $\PE a$.

2. Consider the condition $\NE b$, with $b\colon\onenode{1}\injto\twonodesedgeredge{1}{}{}{}$.
By Construction (1), the rule set $\{\prule\}$ with $\prule = \tuple{x,\twonodesedgeredge{1}{}{}{} \dder \onenode{1},y}$ constitutes the repairing set for $\NE b$, with interfaces~$\onenode{1}$. 
By Construction (2), the rule set $\S_b=\tuple{x,\twonodesedgeredge{1}{2}{}{}\dder\twonodesedge{1}{2}{},y}$ constitutes the repairing set for $\NE b$, with interfaces~$\onenode{1}$. 
The rule set $\S_b$ can be used, e.g., for a repair program $\S_b'\downarrow$ for the condition $\NE (\onenode{1} \injto \twonodesedgeredge{1}{}{}{})$ (see Construction~\ref{const:proper}).
\end{example}

\begin{fact}[compositions of repairing sets]\label{fac:extend}
If $\R_a, \R_a'$ are repairing sets for a basic condition $d$, then $\R_a \cup \R_a'$ is a repairing set for~$d$.
If $\R_a,\R_c$ are  repairing sets for $\PE a$ and $\PE c$, respectively, and  $a = c\circ b$, then $\R_a \cup \R_{ca}$ is a repairing set for $\PE a$, where $\R_{ca} = \{\tuple{b,p,y} \mid \tuple{p,y} \in \R_c\}$.
If $\S_a,\S_c$ are  repairing sets for $\NE a$ and $\NE c$, respectively, and  $a = c\circ b$, then $\S_a \cup \S_{ca}$ is a repairing set for $\NE a$, where $\S_{ca} = \{\tuple{b,p,y} \mid \tuple{p,y} \in \S_c\}$.\end{fact}
\begin{proof}Straightforward.\end{proof}

For proper conditions, a repair program can be constructed.
\begin{theorem}[repair]\label{thm:repair}
For proper conditions, \ignore{maximally preserving, terminating} repair programs can be constructed.
\end{theorem}

\begin{construction}\label{const:proper} For proper conditions $d$ over $A$,  the $A$-program $P_d$  is constructed inductively as follows.
\begin{enumerate}
\item[(1)] For $d = \ctrue$, $P_d = \Skip$.
\item[(2)] For $d = \PE a$, $P_d = \try\R_a$.
\item[(3)] For $d = \NE a$, $P_d = \Rdown{\S_a'}$.
\item[(4)] For $d = \PE(a,c)$, $P_d = P_{\PE a};\tuple{\select(a);P_c;\unselect(a)}$.
\item[(5)] For $d = \PA(a,c)$, $P_d = \Rdown{\tuple{\select(a,\neg c);P_c;\unselect(a)}}$
\end{enumerate}
where $a\colon A\injto C$ with $A\subset C$, $\R_a$ and $\S_a$ are repairing $A$-sets, and $P_c$ is a repair program for $c$ with interfaces $C$.
\end{construction}

{\begin{example}
1. For the constraint $d=\PE(\onenode{1},\NE\twonodesedgeredge{1}{}{}{})$ meaning there exists a node without 2-cycle, i.e., two real edges in opposite direction, $P_d = \tuple{\try \R_a; \tuple{\select(a); \Rdown{\S'_b}; \unselect(a)}}$ where $a\colon\emptyset\injto \onenode{1}$, $b\colon \onenode{1}\injto\twonodesedgeredge{1}{}{}{}$, and $\S_b$ is the repairing set from Example~\ref{ex:Ra}.\ignore{and $\S_b=\{\twonodesedgeredge{1}{2}{}{}\dder\twonodesedge{1}{2}{}\}$.} The program checks whether there exists a node, and if not, it creates one. It selects a node and, if there are two edges in opposite directions, it deletes one. The check of existence is done one time, the deletion as long as possible.

2. For the constraint $d=\PA(\onenode{1},\PE\twonodesedge{1}{}{})$, meaning that, for every node, there exists a real outgoing edge, $P_d = \tuple{\select(a,\neg c);P_{c};\unselect(a)}\downarrow$, is a repair program for $d$, where $a\colon\emptyset\injto \onenode{1}$, $P_c = \try \R_a$ is the repair program for $c=\PE (\onenode{1}\injto\twonodesedge{1}{}{})$, and $\R_a$ is the repairing set from Example~\ref{ex:Ra}.
The repair program selects a node without a real outgoing edge, e.g. the third node from left (see below), applies the rule, and unselects the selected part. Afterwards all nodes possess a real outgoing edge. 

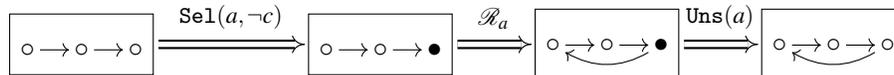
\[\scalebox{1}{
\tikz[node distance=2.5em,inner sep=4pt,outer sep=2pt]{
\node(1)[draw]{$
\tikz[node distance=1em,inner sep=0pt]{\node(space){};
\node(P)[label={[marking]below: \scriptsize },node distance=0.2em,strictly below of=space]{$\circ$};
\node(C)[label={[marking]below: \scriptsize},strictly right of=P]{$\circ$};
\node(D)[label={[marking]below: \scriptsize },strictly right of=C]{$\circ$};
\draw[morphism] (P) edge node[overlay,above]{} (C);
\draw[morphism] (C) -- (D);
}
$};
\node(2)[strictly right of=1,draw,node distance=5em]{$
\tikz[node distance=1em,inner sep=0pt]{\node(space){};
\node(P)[label={[marking]below: \scriptsize },node distance=0.2em,strictly below of=space]{$\circ$};
\node(C)[label={[marking]below: \scriptsize},strictly right of=P]{$\circ$};
\node(D)[label={[marking]below: \scriptsize },strictly right of=C]{${\bullet}$};
\draw[morphism] (P) edge node[overlay,above]{} (C);
\draw[morphism] (C) -- (D);
}
$};
\node(3)[strictly right of=2,draw]{$
\tikz[node distance=1em,inner sep=0pt]{\node(space){};
\node(P)[label={[marking]below: \scriptsize },node distance=0.2em,strictly below of=space]{$\circ$};
\node(C)[label={[marking]below: \scriptsize},strictly right of=P]{$\circ$};
\node(D)[label={[marking]below: \scriptsize },strictly right of=C]{${\bullet}$};
\draw[morphism] (P) edge node[overlay,above]{} (C);
\draw[morphism] (C) edge (D);
\draw[morphism] (D) edge[bend left] node[overlay,above]{} (P);
}
$};
\node(4)[strictly right of=3,draw]{$
\tikz[node distance=1em,inner sep=0pt]{\node(space){};
\node(P)[label={[marking]below: \scriptsize },node distance=0.2em,strictly below of=space]{$\circ$};
\node(C)[label={[marking]below: \scriptsize},strictly right of=P]{$\circ$};
\node(D)[label={[marking]below: \scriptsize },strictly right of=C]{$\circ$};
\draw[morphism] (P) edge node[overlay,above]{} (C);
\draw[morphism] (C) edge (D);
\draw[morphism] (D) edge[bend left] node[overlay,above]{} (P);
}
$};
\draw[derivation] (1) to node[overlay,above]{\footnotesize$\select(a,\neg c)$} (2);
\draw[derivation] (2) to node[overlay,above]{\footnotesize$\R_a$} (3);
\draw[derivation] (3) to node[overlay,above]{\footnotesize$\unselect(a)$} (4);
}}\]
\end{example}}

\begin{proof}[of Theorem \ref{thm:repair}]{\bf  By induction on the structure of the condition.}
Let $d$ be a proper condition and~$P_d$ the program in Construction~\ref{const:proper}. 
(1)  Let $d = \ctrue$. For all triples $\tuple{g,h,i}\in\psem{\Skip}$, $h\circ i\models\ctrue$, i.e., $\Skip$ is a repair program for $\ctrue$.
For (2) and (3) see Lemma \ref{lem:basic}. 

(4) Let $\tuple{g,h,i}\in\psem{P_{\PE a};\tuple{\select(a);P_c;\unselect(a)}}$ (see Figure \ref{fig:illust1}, left). We show that $h\models\PE(a,c)$, i.e., there is some injective morphism $q\colon C\injto H$ such that $p=q\circ a$ and $q\models c$. Let 
 $\tuple{g,h_1,i_1}\in\psem{P_{\PE a}}$ with  $h_1\circ i_1\models \PE a$, 
$\tuple{h_1,h_2,a}\in\psem{\select(a)}$ with $h_1= h_2\circ a$,
$\tuple{h_2,h_3,i_2}\in\psem{P_c}$ with $h_3\circ i_2\models c$,
 $\tuple{h_3,h,a^{-1}}\in\psem{\unselect(a)}$ with $h=h_3\circ a$.
Choose $q=h_3$. Then $h=h_3\circ a= q\circ a$ and, since $P_d$ is $A$-preserving, $q=h_3=h_3\circ i_2\models c$, i.e.,  $h\models\PE(a,c)$. (For $A$-preserving  programs, the interface relation $i_2$ total. Without loss of generality, it is an inclusion.)

(5) Let $\tuple{g,h,i}\in\psem{\Rdown{\tuple{\select(a,\neg c);P_c;\unselect(a)}}}$. We show that $h\models\PA(a,c)=\neg\PE(a,\neg c)$. (see Figure \ref{fig:illust1}, right.) By the semantics of~$\downarrow$, the program $\tuple{\select(a,\neg c);P_c;\unselect(a)}$ is not applicable to the domain of $h_1$. Then there is an injective morphism $h$ such that $h=h_1\circ a$ and $h\models\neg c$, i.e., $h\models\neg\PE(a,\neg c)=\PA(a,c)$. 
\vspace{-1em}
\begin{figure}[h!]
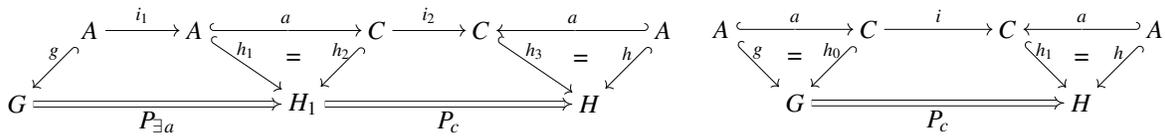

\[\scalebox{0.9}{
\tikz[node distance=4em,shape=rectangle,outer sep=1pt,inner sep=2pt,label distance=-1.25em]{
\node(G){$G$};
\node(A1)[above right of=G]{$A$};
\node(A2)[right of=A1]{$A$};
\node(H1)[node distance=11em, right of=G]{$H_1$};
\node(C1)[above right of=H1]{$C$};
\node(C2)[right of=C1]{$C$};
\node(H)[node distance=11em, right of=H1]{$H$};
\node(A3)[above right of=H]{$A$};
\draw[morphism] (A1) -- node[above](name1){\scriptsize $i_1$} (A2);
\draw[monomorphism] (A2) -- node[above](aname){\scriptsize {$a$}} (C1);
\draw[morphism] (C1) -- node[above]{\scriptsize $i_2$} (C2);
\draw[derivation] (G) -- node(pa)[below]{$P_{\PE a}$} (H1);
\draw[monomorphism] (A1) -- node[above]{\scriptsize $g$} (G);
\draw[monomorphism] (A2) -- node[above]{\scriptsize $h_1$} (H1);
\draw[monomorphism] (C1) -- node[above]{\scriptsize $h_2$} (H1);
\draw[monomorphism] (C2) -- node[above]{\scriptsize $h_3$} (H);
\draw[altmonomorphism] (A3) -- node[above](a2){\scriptsize $a$}(C2);
\draw[monomorphism] (A3) -- node[above](h){\scriptsize $h$} (H);
\draw[derivation] (H1) -- node[below]{$P_c$}(H);
\draw[draw=white] (aname) -- node[overlay]{=} (H1);
\draw[draw=white] (a2) -- node[overlay](tr1){=} (H);}
\quad
\tikz[node distance=4em,shape=rectangle,outer sep=1pt,inner sep=2pt,label distance=-1.25em]{
\node(G){$G$};
\node(A1)[above left of=G]{$A$};
\node(A2)[above right of=G]{$C$};
\node(H)[node distance=11em, right of=G]{$H$};
\node(A3)[above left of=H]{$C$};
\node(A4)[above right of=H]{$A$};
\draw[monomorphism] (A1) -- node[above](i1){\scriptsize $a$} (A2);
\draw[morphism] (A2) -- node[above]{\scriptsize {$i$}} (A3);
\draw[altmonomorphism] (A4) -- node[above](i3){\scriptsize $a$} (A3);
\draw[derivation] (G) -- node[below]{$P_c$} (H);
\draw[monomorphism] (A1) -- node[above](a){\scriptsize $g$} (G);
\draw[monomorphism] (A2) -- node[above](a){\scriptsize $h_0$} (G);
\draw[monomorphism] (A3) -- node[above](a){\scriptsize $h_1$} (H);
\draw[monomorphism] (A4) -- node[above](a){\scriptsize $h$} (H);
\draw[draw=white] (i1) -- node[overlay]{=} (G);
\draw[draw=white] (i3) -- node[overlay](tr1){=} (H);
}}\]
\caption{\label{fig:illust1} Illustration of the proof}
\end{figure}
\end{proof}
In the following, we look for properties of the constructed repair programs. Whenever a condition graph requires the non-existence (existence) of a certain subgraph, there is no non-deleting (non-adding) repair program that repairs all graphs. Therefore, we look for minimally deleting or maximally preserving repair programs.
\ignore{To get some ranking between the possible resulting graphs, we consider \emph{maximally preserving} repair programs, i.e. where items are preserved whenever possible. Whenever a graph satisfies a constraint, but the negation is required, then the graph cannot be repaired without deletions. In this case, at least as necessary items should be deleted. }

\begin{definition}[maximally preserving repair]\label{def:mpres}
A repair program $P_d$ for a proper condition $d$ is \emph{maximally preserving}, if, for all transformations $t\colon G\dder_{{P_d},g,h}H$, \[\pres(P_d,t)\geq\size(G)-\Delta(g,d)\] 
where, for a transformation $t$ via $P_d$, $\pres(P_d,t)$ denotes the number of preserved items by $t$, i.e., the items in the domain of the partial track morphisms of $t$, and $\Delta(g,d)$ denotes the maximum number of necessary deletions. 
\footnotetext{For a set $S$, $|S|$ denotes the number of elements.}
Given an injective morphism $g\colon A\injto G$ and a proper condition $d$, $\Delta(g,d)$ is
defined inductively as follows: $\Delta(g,\ctrue) = 0$, $\Delta(g,\PE a)=~0$,  
\[\begin{array}{llclllll}
(1)&\Delta(g,\NE a)&=&\sum_{g'\in\Ext(g)} (1+\dang(g'))\\
(2)&\Delta(g,\PE(a,c))&=&\max_{g' \in\Ext(g)} \Delta(g',c)\\
(3)&\Delta(g,\PA(a,c))&=& \sum_{g'\in\Ext(g)}(\Delta(g',c))\\
\end{array}\]
where \ignore{$\Mor(A,G)$ denotes the set of all injective morphisms $g\colon A\injto G$,} $\Ext(g)=\{g'\colon C\to G\mid g' \circ a=g\}$, and $\dang(g')$ denotes the maximum number of dangling edges at $g'$, i.e. $\dang(g')=0$ if there is some edge in $g'(C-A)$ and $\max_{v \in (C-A)}\inc(v)$, otherwise, where, for a node $v$, $\inc(v)$ denotes the number of edges incident to $v$.
\end{definition}

\begin{remark}Given a morphism $g\colon A\injto G$ and a proper condition $d$, we determine the maximal number of necessary deletions $\Delta(g,d)$. This is zero if the condition is $\ctrue$ or of the form $\PE a$. For a condition of the form $\NE a$, we consider all morphisms $g'\in\Ext(g)$, and sum up the number of deletions.
For a condition of the form $\PE(a,c)$, we consider all morphisms $g'\in\Ext(g)$ and build the maximum of all $\Delta(g',c)$. For a condition of the form $\PA(a,c)$, we consider all morphisms $g'\in\Ext(g)$ and sum up the  number of necessary deletions for the condition $c$ at that position, i.e. $\Delta(g',c)$.
\end{remark}

\begin{fact}Repair program based on (1), in general, are neither terminating nor maximally preserving.\end{fact}
\begin{proof}For the condition $\PA(\onenode{1},\PE\twonodesedge{1}{}{})$, $\tuple{\onenode{1}\dder\twonodesedge{1}{}{}}\downarrow$ is a repair program based on (1) not creating cycles. The same holds for $\tuple{\onenode{1}\dder\twonodesedge{1}{}{},\NE\twonodesedge{1}{}{}}\downarrow$. Both programs are not terminating. A~rule $C\dder A$ deletes $[C-A]$ items, although only one item has to be deleted, i.e., in general it is not maximally preserving. \end{proof}

\begin{lemma}[program properties]\label{fac:mpres}
The repair program based on Construction \ref{const:basic}(2) is terminating and maximally preserving.
\end{lemma}

\begin{proof}The termination of the repair program based on (2)  is shown in \cite{Habel-Sandmann18a}. The maximal preservation of the repair program $P_d$ based on (2) is shown by induction of the length of transformations: We show that, for all transformations $t\colon G\dder_{P_d,g,h}~H$, 
\[\pres(P_d,t)\geq\size(G)-\Delta(g,d).\]

Let $d$ be a proper condition, $P_d$ the repair program for $d$, and $t\colon G\dder_{P_d,g,h}~H$ a transformation. 
\begin{enumerate}
\item[(1)] For $d = \ctrue$, $P_d=\Skip$, and $\pres(\Skip,t)=\size(G)$.
\item[(2)] For $d = \PE a$, $P_d=\try \R_a$, and $\pres(\try \R_a,t)=\size(G)$.

\item[(3)] For $d = \NE a$. $P_d=\Rdown{\S'_a}$. (a) If $g\models d$, then $\pres(\Rdown{\S'_a},t)=\size(G)$. (b) If $g\not\models d$, then the transformation $G\dder_{\S'_ a\downarrow,g,h}^{+}H$ is of the form $G\dder_{\S'_a,g,g_1} G_1\dder_{\S'_ a\downarrow,g_1,h}H$ where $t_1$ denotes the transformation starting with $G_1$. Then, for all $g'\in\Ext(g)$,  (*) $\size(G_1)=\size(G)-\Delta(g',d)$ with $\Delta(g',d)=1+\dang(g')$. By definition of $\Delta$, 
(**) $\Delta(g,d)=\Delta(g',d)+\Delta(g_1,d)$.

\[\begin{array}{lcll}
\pres(P_d,t)
&\geq & \pres(P_d,t_1)\\
&\geq&\size(G_1)-\Delta(g_1,d)&\mbox{(induction hypothesis)}\\
&=&\size(G)-\Delta(g,d)&\mbox{((*), (**))}\\
\end{array}\]

\item[(4)] For $d=\PE(a,c)$, $P_d = P_{\PE a};\tuple{\select(a);P_c;\unselect(a)}$.
(a) If $g\models d$, then $\Pres(P_d,G)=\size(G)$. (b) If $g\not\models d$, then  $G\dder_{P_d,g.h}H$ is of the form  $G\dder_{P_{\PE a},g,g'} G_1\dder_{P'_c,g',h}H$. Then $\size(G_1)\geq\size(G)$ and, for all $g'\in\Ext(g)$,  (**)~$\Delta(g,d)=\Delta(g',c)$. 

\[\begin{array}{lcll}
\pres(P_d,t)
&=&\pres(P'_c,t_1)&\mbox{}\\
&\geq&\size(G_1)-\Delta(g',c)&\mbox{(induction hypothesis)}\\
&=&\size(G)-\Delta(g,d)&\mbox{((*), (**))}\\
\end{array}\]

\item[(5)] For $d = \PA(a,c)$, $P_d = \Rdown{\tuple{\select(a,\neg c);P_c;\unselect(a)}}$. (a) If $g\models d$, then $\Pres(P_d,G)=\size(G)$. (b)~If $g\not\models d$, then $G\dder_{P_d,g,h}H$ is of the form  $G\dder_{P'_c,g,g_1} G_1\dder_{P_d,g_1,h}H$ where $P'_c$ denotes the program without iteration. 
If $c$ is of the form $\PE(a',c')$, then, $\size(G_1)\geq\size(G)$ and, for every $g'\in\Ext(g)$, $\Delta(g',c)=~0$. 
If $c$ is of the form $\NE a'$, then, for every $g'\in\Ext(g)$, $\size(G_1)=\size(G)-\Delta(g',c)$ as in Case (3). 
Thus,  (*)~$\size(G_1)\geq\size(G)-\Delta(g',c)$. By definition of $\Delta$,  
(**)~$\Delta(g,d)=\Delta(g_1,d)+\Delta(g',c)$.

\[\begin{array}{lcll}
\pres(P_d,t)
&\geq& \pres(P_d,t_1)\\
&\geq&\size(G_1)-\Delta(g_1,d)&\mbox{(induction hypothesis)}\\
&\geq&\size(G)-\Delta(g,d) &\mbox{((*), (**))}\\
\end{array}\]
\end{enumerate}
This completes the inductive proof.
\end{proof}

\section{Rule-based  repair}\label{sec:rb-repair}

A rule-based program is a program based on a set of rules equipped with the dangling-edges operator, context, application condition, and interface.
\begin{definition}[rule-based programs]\label{def:rule-based}Given a set of rules $\R$, a program is \emph{$\R$-based}, if all rules in the program are rules in $\R$ equipped with dangling-edges operator, context, application condition, and interface. Additionally, the empty program $\Skip$ is $\R$-based. 
\end{definition}
\vspace{1em}
\begin{example}\label{ex:rb-repair}The rule
$\build = \brule{\;\embedtikz{
\node[waypoint,label={[below, yshift=-0.2cm]\tiny 1}] at (1,-2) (v1) {};
\node[waypoint,label={[below, yshift=-0.2cm]\tiny 2}] at (2,-2) (v2) {};
}}
{\;\embedtikz{
\node[waypoint,label={[below, yshift=-0.2cm]\tiny 1}] at (1,-2) (v1) {};
\node[waypoint,label={[below, yshift=-0.2cm]\tiny 2}] at (2,-2) (v2) {};
}}
{\;\embedtikz{
\node[waypoint,label={[below, yshift=-0.2cm]\tiny 1}] at (1,-2) (v1) {};
\node[waypoint,label={[below, yshift=-0.2cm]\tiny 2}] at (2,-2) (v2) {};
\draw[trackgray] (v1) edge (v2);\draw[trackwhite] (v1) edge (v2);
}\;\;}$\vspace{0.1cm}
equipped with the context\\ $\twowaypointsindex{1}{2} \injto \onelokindex{1}{2}$ and 
the application condition $\NE \lokontrackindex{1}{2} \wedge \PE \onelokindex{1}{2}$
yields to the $\{\build \}$-based program $ \try \build2$ with  
$\build2 = \tuple{\onelokindex{1}{2} \injlto \onelokindex{1}{2} \injto \lokontrackindex{1}{2},\NE \lokontrackindex{1}{2} \wedge \PE \onelokindex{1}{2}}$.
\end{example}

The construction of an $\R$-based repair program for a proper condition $d$ is based on the following idea (see Figure~\ref{fig:based}).
\begin{itemize}
\item[(1)] Take a repair program for the condition $d$ (Theorem \ref{thm:repair}).
\item[(2)] Try to refine the rules of the repairing sets by equivalent transformations via $\R$.
\item[(3)] Transform the transformations into equivalent $\R$-based programs (Theorem \ref{thm:trafo}).
\item[(4)] Replace each repairing set in $P_d$ by the equivalent $\R$-based program (Theorem \ref{thm:rb-repair}).
\end{itemize}

\begin{figure}[h!]
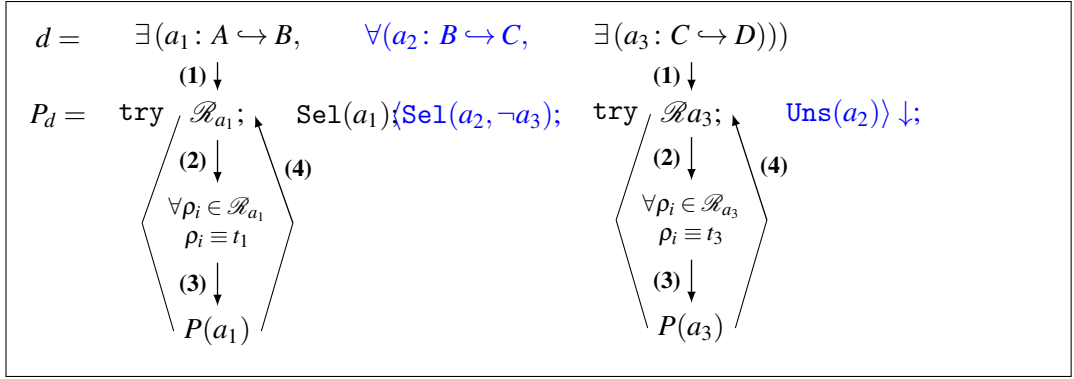

\[\scalebox{1}{
\tikz[node distance=1.5em,inner sep=3pt,outer sep=1pt]{
\node(a1) {$\PE (a_1\colon A \injto B, $};
\node(a2)[strictly right of=a1] {\blue $\PA (a_2\colon B \injto C,$};
\node(a3)[strictly right of=a2] {$\PE (a_3\colon C \injto D)))$};
\node(d)[left of=a1,node distance=5.5em]{$d = $};

\node(p1) [node distance=1em,strictly below of=a1] {$\R_{a_1};$};
\node(tryp1) [left of=p1,node distance=2.5em]{$\try$};
\node(seln) [right of=p1,node distance=4.5em] {$\select(a_1);$};
\node(p2) [node distance=1em,xshift=1em,strictly below of=a2] {\blue $\langle \select(a_2,\neg a_3);$};
\node(p3) [node distance=1em,strictly below of=a3] {$\R{a_3};$};
\node(tryp3) [left of=p3,node distance=2.5em]{$\try$};
\node(uns) [strictly right of=p3] {\blue $\unselect(a_2)\rangle\downarrow;$};
\node(Pd)[left of=p1,node distance=5.5em]{$P_d = $}; 

\node(t1) [strictly below of=p1] {\footnotesize \begin{tabular}{c}$\PA \prule_i \in \R_{a_1}$ \\ $\prule_i \equiv t_1$\end{tabular}};
\node(t2) [strictly below of=p2] {\phantom{$A \DSdder_{\R}^{*} C$}};

\node(t3) [strictly below of=p3] {\footnotesize \begin{tabular}{c}$\PA \prule_i \in \R_{a_3}$ \\ $\prule_i \equiv t_{3}$\end{tabular}};

\node(pa1) [strictly below of=t1] {$P(a_1)$};
\node(pa2) [strictly below of=t2] {\phantom{$P(a_2)$}};
\node(pa3) [strictly below of=t3] {$P(a_3)$};


\draw[arrow] (t1) to node [left] {\footnotesize \textbf{(3)}} (pa1);
\draw[arrow] (t3) to node [left] {\footnotesize  \textbf{(3)}} (pa3);
\draw[arrow] (p1) to node [left] {\footnotesize \textbf{(2)}} (t1);
\draw[arrow] (p3) to node [left] {\footnotesize \textbf{(2)}} (t3);

\draw[arrow] (a1) to node [left] {\footnotesize \textbf{(1)}} (p1);
\draw[arrow] (a3) to node [left] {\footnotesize \textbf{(1)}} (p3);

\draw[-] (t1.west) to (p1.west);
\draw[arrow] (t1.east) to node[right] {\footnotesize \textbf{(4)}} (p1.east);
\draw[-] (t3.west) to (p3.west);
\draw[arrow] (t3.east) to node[right] {\footnotesize \textbf{(4)}} (p3.east);

\draw[-] (pa1.west) to (t1.west);
\draw[-] (pa1.east) to node[right]{} (t1.east);
\draw[-] (pa3.west) to (t3.west);
\draw[-] (pa3.east) to node[right]{} (t3.east);


\node [strictly left of=d,node distance=0.3em](h2){};
\node [strictly above of=h2,node distance=0.5em](ol){};
\node [strictly right of=a3,node distance=9em] (h1){};
\node [strictly below of=h1,node distance=11em](ur){};
\draw (ur) rectangle (ol);
}}\]
\caption{\label{fig:based}Construction of an $\R$-based repair program}
\end{figure}

In the following, we introduce the main notion of compatibility, saying that, for all rules of the repairing sets of the repair program for $d$, there are equivalent transformations via the rule set.

\begin{definition}[equivalence]
Two transformations $t,t'$ from $G$ to $H$ are \emph{equivalent}, denoted $t \equiv t'$, if for each extension form $G^*$ to $H^*$ there is an extension of $t'$ from $G^*$ to $H^*$, and vice versa.
\end{definition}

\ignore{
\begin{remark}
Instead of equivalent transformations, one may consider approximating transformations:
Two transformations $t,t'$ starting from the same graph are
\emph{replaceable} if for each extension of $t$, there is one for $t'$,
and vice versa. A transformation $t'$ approximates $t$ w.r.t.
$\ac$, denoted by $t\leq_{\ac}t'$, if the transformations $t,t'$ are
replaceable and, for all the triples $\tuple{g_1,h_1,i_1}\in\psem{t_1}$ and
$\tuple{g_2,h_2,i_2}\in\psem{t_2}$, $h_1\circ i_1\models\ac\shortiff h_2\circ
i_2\models \ac$.
\end{remark}
}

\begin{definition}[compatibility]\label{def:comp}Let $d$ be a proper condition. A set of rules $\R$ is \emph{$d$-compatible} (w.r.t. a repair program $P_d$) if, for all rules in the repairing sets of $P_d$, there are {\emph equivalent} transformations via $\R$. In particular, if $\R=\{\prule\}$, we also say that $\prule$ is $d$-compatible.
\end{definition}

\begin{example}
Let $\NoTwo=\NE(\emptyset\injto\scalebox{0.7}{\embedtikz{
\node[waypoint,label={[below left, yshift=-0.2cm]\tiny}] at (1,-2) (v1) {};
\node[waypoint,label={[below right, yshift=-0.2cm]\tiny}] at (2,-2) (v2) {};
\draw[trackgray] (v1) edge (v2);\draw[trackwhite] (v1) edge (v2);
\draw[arrow] (v1) edge[bend angle=90, bend left,min distance=1em] node {\nlok{}} (v2);
\draw[arrow] (v1) edge[bend angle=90, bend right,min distance=1em] node {\nlok{}} (v2);
}})$.  
Then $\{\Delete\}$ is a repairing set for $\NoTwo$, and there is a transformation
\scalebox{0.7}{$\embedtikz{
\node[waypoint,label={[below left, yshift=-0.2cm]\tiny}] at (1,-2) (v1) {};
\node[waypoint,label={[below right, yshift=-0.2cm]\tiny}] at (2,-2) (v2) {};
\draw[trackgray] (v1) edge (v2);\draw[trackwhite] (v1) edge (v2);
\draw[arrow] (v1) edge[bend angle=90, bend left,min distance=1em] node {\nlok{}} (v2);
\draw[arrow] (v1) edge[bend angle=90, bend right,min distance=1em] node {\nlok{}} (v2);
}\DSdder_{\Delete2}
\embedtikz{
\node[waypoint,label={[below left, yshift=-0.2cm]\tiny}] at (1,-2) (v1) {};
\node[waypoint,label={[below right, yshift=-0.2cm]\tiny}] at (2,-2) (v2) {};
\draw[trackgray] (v1) edge (v2);\draw[trackwhite] (v1) edge (v2);
\draw[arrow] (v1) edge[bend angle=90, bend right,min distance=1em] node {\nlok{}} (v2);
}$} 
via the $\Delete$-based program  
$\Delete2:\quad$\scalebox{0.7}{$\acrule{\embedtikz{
\node[waypoint,label={[below, yshift=-0.2cm]\tiny 1}] at (1,-2) (v1) {};
\node[waypoint,label={[below, yshift=-0.2cm]\tiny 2}] at (2,-2) (v2) {};
\draw[trackgray] (v1) edge (v2);\draw[trackwhite] (v1) edge (v2);
               \draw[arrow] (v1) edge[bend angle=90, bend left,min distance=1em] node {\nlok{}} (v2);
}}
{\embedtikz{
\node[waypoint,label={[below, yshift=-0.2cm]\tiny 1}] at (1,-2) (v1) {};
\node[waypoint,label={[below, yshift=-0.2cm]\tiny 2}] at (2,-2) (v2) {};
\draw[trackgray] (v1) edge (v2);\draw[trackwhite] (v1) edge (v2);
                }}
                {\PE\embedtikz{
\node[waypoint,label={[below, yshift=-0.2cm]\tiny}] at (1,-2) (v1) {};
\node[waypoint,label={[below, yshift=-0.2cm]\tiny}] at (2,-2) (v2) {};
\draw[trackgray] (v1) edge (v2);\draw[trackwhite] (v1) edge (v2);
               \draw[arrow] (v1) edge[bend angle=90, bend left,min distance=1em] node {\nlok{}} (v2);
                \draw[arrow] (v1) edge[bend angle=90, bend right,min distance=1em] node {\nlok{}} (v2);
}\;}$}$\downarrow$.
The rule set $\{\Move\}$ is a repairing set for $\NE (\emptyset \injto \scalebox{0.7}{\tikz[]{
\node[waypoint,label={[below, yshift=-0.2cm]\tiny }] at (1,-2) (v1) {};
\node[waypoint,label={[below, yshift=-0.2cm]\tiny }] at (2,-2) (v2) {};
\node[waypoint,label={[below, yshift=-0.2cm]\tiny }] at (3,-2) (v3) {};
\draw[trackgray] (v1) edge (v2);\draw[trackwhite] (v1) edge (v2);
\draw[trackgray] (v2) edge (v3);\draw[trackwhite] (v2) edge (v3);
 \draw[arrow] (v1) edge[bend angle=90, bend left,min distance=1em] node {\nlok{}} (v2);
  \draw[arrow] (v1) edge[bend angle=-90, bend left,min distance=1em] node {\nlok{}} (v2);
}})$. By Fact \ref{fac:extend}, the rule set $\{\Move,\Delete\}$ is a repairing set for $\NoTwo$. 

\ignore{
there is a transformation $\scalebox{0.7}{\embedtikz{
\node[waypoint,label={[below, yshift=-0.2cm]\tiny }] at (1,-2) (v1) {};
\node[waypoint,label={[below, yshift=-0.2cm]\tiny }] at (2,-2) (v2) {};
\node[waypoint,label={[below, yshift=-0.2cm]\tiny }] at (3,-2) (v3) {};
\draw[trackgray] (v1) edge (v2);\draw[trackwhite] (v1) edge (v2);
\draw[trackgray] (v2) edge (v3);\draw[trackwhite] (v2) edge (v3);
 \draw[arrow] (v1) edge[bend angle=90, bend left,min distance=1em] node {\nlok{}} (v2);
 \draw[arrow] (v1) edge[bend angle=90, bend right,min distance=1em] node {\nlok{}} (v2);
}}\DSdder_{\Move2}
{\embedtikz{
\node[waypoint,label={[below, yshift=-0.2cm]\tiny }] at (1,-2) (v1) {};
\node[waypoint,label={[below, yshift=-0.2cm]\tiny }] at (2,-2) (v2) {};
\node[waypoint,label={[below, yshift=-0.2cm]\tiny }] at (3,-2) (v3) {};
\draw[trackgray] (v1) edge (v2);\draw[trackwhite] (v1) edge (v2);
\draw[trackgray] (v2) edge (v3);\draw[trackwhite] (v2) edge (v3);
\draw[arrow] (v2) edge[bend angle=90, bend left,min distance=1em] node {\nlok{}} (v3);
\draw[arrow] (v1) edge[bend angle=90, bend right,min distance=1em] node {\nlok{}} (v2);
}}$ 
via the $\{\Move,\Delete\}$-based program \\
\[\Move2:\quad\scalebox{0.7}{$\tuple{\embedtikz{
\node[waypoint,label={[below, yshift=-0.2cm]\tiny }] at (1,-2) (v1) {};
\node[waypoint,label={[below, yshift=-0.2cm]\tiny }] at (2,-2) (v2) {};
\node[waypoint,label={[below, yshift=-0.2cm]\tiny }] at (3,-2) (v3) {};
\draw[trackgray] (v1) edge (v2);\draw[trackwhite] (v1) edge (v2);
\draw[trackgray] (v2) edge (v3);\draw[trackwhite] (v2) edge (v3);
 \draw[arrow] (v1) edge[bend angle=90, bend left,min distance=1em] node {\nlok{}} (v2);
}\DSdder
\embedtikz{
\node[waypoint,label={[below, yshift=-0.2cm]\tiny }] at (1,-2) (v1) {};
\node[waypoint,label={[below, yshift=-0.2cm]\tiny }] at (2,-2) (v2) {};
\node[waypoint,label={[right, yshift=-0.2cm]\tiny }] at (3,-2) (v3) {};
\draw[trackgray] (v1) edge (v2);\draw[trackwhite] (v1) edge (v2);
\draw[trackgray] (v2) edge (v3);\draw[trackwhite] (v2) edge (v3);
\draw[arrow] (v2) edge[bend angle=90, bend left,min distance=1em] node {\nlok{}} (v3);
},
\PE\embedtikz{
\node[waypoint,label={[below, yshift=-0.2cm]\tiny}] at (1,-2) (v1) {};
\node[waypoint,label={[below, yshift=-0.2cm]\tiny}] at (2,-2) (v2) {};
\node[waypoint,label={[below, yshift=-0.2cm]\tiny}] at (3,-2) (v3) {};
\draw[trackgray] (v1) edge (v2);\draw[trackwhite] (v1) edge (v2);
                \draw[trackgray] (v2) edge (v3);\draw[trackwhite] (v2) edge (v3);
               \draw[arrow] (v1) edge[bend angle=90, bend left,min distance=1em] node {\nlok{}} (v2);
                \draw[arrow] (v1) edge[bend angle=90, bend right,min distance=1em] node {\nlok{}} (v2);
}
\;}$}\downarrow.\]
}
\end{example}

{{ \begin{example}[no $\{\Move,\Delete\}$-based repair]
Consider the constraint $\Station=\PE\scalebox{0.5}{$\rightstation$}$ (there exists a train station). Whenever the start graph has no station, then no station can be created by a $\{\Move,\Delete\}$-based program. The reason is that the labels of the constraint do not occur in the right-hand sides of the rules.
\end{example}
} 
}

In the case of $\R$ is $d$-compatible w.r.t. $P_d$, for all rules in the repair program, there are transformations via $\R$. These transformations via $\R$ can be transformed into $\R$-based programs.

\begin{theorem}[from transformations to rule-based programs]\label{thm:trafo} 
For every transformation $t\colon G\der_\R H$, there is a $\R$-based program $P(t)$ such that $t\equiv P(t)$. 
\end{theorem}

\begin{construction}Let $t:G\der_\R H$ be a transformation.
For direct transformations $t\colon G\dder_{\prule,g,h} H$ via a rule $\prule=\tuple{x,p,\ac,y}$ with interfaces $X$ and $Y$, let $P(t):=\tuple{\select(g'\circ x,\ac');\bar{\prule}';\unselect(h'\circ y)}$ be the rule with left interface $g'\circ x$, $\bar{\prule}=G\dder H$ be the rule $\prule$ equipped with context, $\ac'=\Shift(g',\ac)$ the left application condition for $\bar{\prule}$, and $h'\circ y$ the right interface. For transformations $t\colon G=G_0\dder_{\R}^{n+1} G_{n+1}=H$, with $t_1 \colon G_0\dder_\R^n G_n$, and $t_2\colon G_n \dder_\prule G_{n+1}$, $P(t):= \tuple{P(t_1); P(t_2)}$.
\end{construction}


\begin{proof}Let $t:G\der_\R H$ be a transformation. By construction, $P(t)$ is $\R$-based. We show that there is a transformation $G\dder_{P(t)} H$. For one-step transformations, by construction, $t\equiv P(t)$. For $n+1$-step transformations, by  induction hypothesis, $t_1\equiv P(t_1)$ and $t_2\equiv P(t_2)$. Then $P(t):=\tuple{P(t_1);P(t_2)}$ is a program with $t\equiv P(t)$.\end{proof}

\ignore{Given a condition $d$ over $A$, and repair transformation for $d$, we construct a repair program inductively as follows:
If the transformation step yields to the satisfaction of $d$, we equip the rule with an application condition $\ac'$, such that the rule  is applicable if there is a violation of $d$, and the transformation step is {\red $d$-guaranteeing}.
If the transformation step does not directly yield to the satisfaction , we construct the $\neg d$-preserving application condition.
By the left interface $\tr\circ x$, the graph $A$ is fixed for the next rule application. 
It says that at this position (and no other) the rule shall be applied.}

\begin{theorem}[rule-based repair]\label{thm:rb-repair} For every proper condition $d$, every repair program $P_d$ for $d$, and every rule set $\R$ $d$-compatible w.r.t. $P_d$, there is an $\R$-based repair program for $d$.\end{theorem}

\begin{construction}Let $P'_d=P_d[\repl]$ where $\repl$ is constructed as follows:
By assumption, for the rules in the repairing $A$-set, there are equivalent \ignore{$A$-}transformations via $\R$. For these transformations, there are equivalent $\R$-based programs with interfaces (Theorem \ref{thm:trafo}). The mapping $\repl$ replaces the repairing $A$-sets by equivalent $\R$-based programs with interfaces $A$. 
\end{construction}

\begin{proof} By assumption, for all \ignore{$A$-}rules in the repairing $A$-sets of $P_d$, there are equivalent \ignore{$A$-}transformations via $\R$. By Theorem \ref{thm:trafo}, the transformations can be transformed into equivalent $\R$-based repair programs. This yields a mapping $\repl$ which replaces the repairing $A$-sets by equivalent $\R$-based programs with interfaces $A$. 
By the Leibniz's replacement principle\ignore{Leibniz's law}, the repair programs $P_d$ and $P_d[\repl]$ are equivalent. Thus, $P_d[\repl]$ is an $\R$-based repair program for $d$.
\end{proof}

To get maximally preserving rule-based repair programs, we have to assume, that our input rule set is maximally preserving, as well. If a rule set is maximally preserving, then the number of deleted items is minimal. For non-deleting rule sets, the number of deleted elements is $0$, and the graph can be preserved. If a rule set is deleting, we delete edges instead of nodes, whenever possible, since it is more costly to delete nodes than edges.

\begin{fact}[program properties]
Rule-based repair programs based on Construction \ref{const:basic}(1), in general, are neither terminating nor maximally preserving. Rule-based repair programs based on Construction \ref{const:basic}(2) are terminating and maximally preserving.
\end{fact}
\begin{proof}
The statements follow immediately form the corresponding statements in Fact \ref{fac:mpres}.
\end{proof}

\ignore{
{\input{Examples/4-chains-compatible}}
}

\section{Related concepts}\label{sec:related}

In this section, we present some related concepts on rule-based graph repair. For the related problem of model repair, there is a wide variety of different approaches. For a more sophisticated survey on different model repair techniques, and a feature-based classification of these approaches, see \textbf{Macedo et al.} \cite{Macedo-etal17a}.

{\bf Rule-based repair.} The notion rule-based repair is used in different meanings. In most cases \cite{Nebras-etal17a,Habel-Sandmann18a}, a rule set is derived from a set of constraints and a repair algorithm/program is constructed from the rule set. In this paper, a rule set and a condition are given as input and a repair program is constructed from the rule set.

In {\bf Nassar et al. 2017} \cite{Nebras-etal17a}, a rule-based approach to support the modeler in automatically trimming and completing EMF models is presented. For that, repair rules are automatically generated from multiplicity constraints imposed by a given meta-model.
\ignore{The approach to model repair consists of two activities: 
(1) The meta-model of a given language is translated to a rule-based model transformation system (repair rules). 
(2) A potentially invalid model is repaired yielding a valid model. The repair algorithm uses the generated transformation system.}
The rule schemes are carefully designed to consider the EMF model constraints defined in \cite{Biermann-Ermel-Taetzer12a}.

In {\bf Habel and Sandmann 2018} \cite{Habel-Sandmann18a}, given a {proper} condition, we derive a rule set from the condition $d$ and construct a repair program using this rule set. The repair program is required to repair all graphs.
In this paper, we use programs with interface \cite{Pennemann09a} with selection and unselection of parts, instead of markings as in \cite{Habel-Sandmann18a}. {\red For simple cases and illustration purposes, marking may be an alternative. For conditions with large nesting depth, the morphism-based concept is more convenient, the marking of elements requires an additional marking for each nesting depth.} 

In {\bf Schneider et al. 2019} \cite{Schneider-etal19a}, a logic-based incremental approach to graph repair is presented, generating a sound and complete (upon termination) overview of least changing repairs. The graph repair algorithm takes a graph and a first-order (FO) graph constraint as inputs and returns a set of graph repairs. Given a condition and a graph, they compute a set of symbolic models, which cover the semantics of a graph condition.

\begin{table}[htp]
\[\scalebox{1}{
\begin{tabular}{|l|c|c|c|c|} \hline
& Schneider et al. \cite{Schneider-etal19a} & Nassar et al. \cite{Nebras-etal17a} & Habel et al. \cite{Habel-Sandmann18a} & this work \\ \hline
input &  FO condition & {EMF model} & FO condition &  FO condition \\ 
& \& graph & & & \& rule set \\ \hline 
output & repair algorithm & repair rules & repair program & repair program \\ 
& & \& valid model & & \\ \hline
correctness & + & + & + & + \\ \hline
dynamic & + & - & - & -\\ \hline 
termination & - & (+) & + & {(+)}\\ \hline
\end{tabular}
}\]
\caption{\label{fig:table}Overview of selected repair approaches}
\end{table}

All approaches are proven to be \textbf{correct}, i.e. the repair (programs) yield to a graph satisfying the condition. 
In Schneider et al. \cite{Schneider-etal19a} the delta-based repair algorithm takes the graph update history explicitly into account, i.e. the approach is \textbf{dynamic}. In contrast, our approach is static, i.e., we first construct a ($\R$-based) repair program, then apply this program to an arbitrary graph. 
In Schneider et al. \cite{Schneider-etal19a}, the program does not \textbf{terminate}, if the repair updates trigger each other ad infinitum. If we choose the repairing set accordingly, we get a terminating repair program.

\ignore{In both approaches, the condition is transformed into a {\red \textbf{conjunctive normal form} \cite{Pennemann09a}. 
The difference is, that our normal form is then transformed into a condition with alternating quantifiers, in \cite{Schneider-etal19a}, additionally, unsatisfiable conditions are removed relying on an ``oracle'' which decides the satisfiability problem.}
W.l.o.g., both approaches assume, that each condition of the form $\PE (A \injto C)$ (or $\NE (A \injto C)$) has an inclusion morphism, such that $A$ is a proper subgraph of $C$.}

In {\bf Taentzer et al. 2017} \cite{Taentzer-etal17a}, a designer can specify a set of so-called change-preserving rules, and a set of edit rules.
Each edit rule, which yields an inconsistency, is then repaired by a set of repair rules.
The construction of the repair rules is based on the complement construction.
It is shown, that a consistent graph is obtained by the repair program, provided that each repair step is sequentially independent from each following edit step, and each edit step can be repaired. 
The repaired models are not necessarily as close as possible to the original model.


In {\bf Cheng et al. 2018} \cite{Cheng-etal18a}, a rule-based approach for graph repair is presented. Given a set of rules, and a graph, they use this set of rules, to handle different kinds of conditions, i.e., incompleteness, conflicts and redundancies. 
The rules are based on seven different operations not defined in the framework of the DPO-approach.
They look for the ``best'' repair based on the ``graph edit distance''.

\section{Conclusion}\label{sec:conclusion}
The repair programs are formed from rules derived from the given proper condition.
They were constructed to be maximally preserving, i.e. to preserve nodes as much as possible.
Additionally, we have considered rule-based repair where the repair programs are constructed from a given set of small rules and a given condition. 
Based on the repair program for proper conditions, we have constructed a rule-based repair program for proper conditions provided that the given rule set is compatible with the repairing sets of the original program. 

Summarizing, we have constructed
\begin{enumerate}
\item[(1)] repair programs for proper conditions (Theorem \ref{thm:repair}),
\item[(2)] rule-based programs from transformations (Theorem \ref{thm:trafo}),
\item[(3)] rule-based repair programs for proper conditions provided that the given rule set is compatible with the repairing sets of the original program (Theorem \ref{thm:rb-repair}).
\end{enumerate}

Further topics are rule-based repair programs for all satisfiable conditions, for typed attributed graphs and EMF-models, i.e., typed, attributed graphs satisfying some constraints, and an implementation.
\ignore{
\[
\scalebox{0.8}{
\tikz[node distance=5em,inner sep=3pt]{
\node(M)[label={above right:}] {\begin{tabular}{c}invalid  Model $M$\end{tabular}};
\node(G)[below of=M]{invalid Graph $G$};
\node(H)[right of=G,node distance=14em]{\begin{tabular}{c}valid  Graph $H$ \end{tabular}};
\node(M')[right of=M,node distance=14em] {\begin{tabular}{c}valid  Model $N$\end{tabular}};
\node(h)[right of=M,node distance=6em]{};
\node(MM)[above of=h,node distance=4em]{\begin{tabular}{c}Meta-model \& \\ constraint\end{tabular}};
\draw[-] (MM) to node[above,sloped,midway] {\scriptsize conforms} (M');
\draw[-] (MM) to node[above,sloped,midway] {\scriptsize } (M);
\draw[->,dashed] (M) to node[left] {\begin{tabular}{c} e.g. \\ \cite{Biermann-Ermel-Taetzer12a}\end{tabular}} (G);
\draw[->] (M) to node[above] {$\R$-based} node[below] {model repair} (M');
\draw[->] (G) to node[above] {$\R$-based} node[below] {graph repair} (H);
\draw[->,dashed] (H) to node[right] {} (M');
}}\]
}
The aim is to represent the structure of a meta model as graph-like structure, and OCL constraints as nested graph conditions, and then use the ($\R$-based) graph repair for ($\R$-based) model repair.


\paragraph{\normalfont\textbf{Acknowledgements.} We are grateful to Marius Hubatschek, Jens Kosiol, Nebras Nassar, and the anonymous reviewers for their helpful comments to this paper.}

\bibliographystyle{eptcs}
\bibliography{bib/lit-GRAGRA,bib/lit-LOGIK,bib/lit-model-repair,bib/lit-MCheck,bib/lit-ocl,bib/lit-PN}

\begin{thebibliography}{10}
\providecommand{\bibitemdeclare}[2]{}
\providecommand{\surnamestart}{}
\providecommand{\surnameend}{}
\providecommand{\urlprefix}{Available at }
\providecommand{\url}[1]{\texttt{#1}}
\providecommand{\href}[2]{\texttt{#2}}
\providecommand{\urlalt}[2]{\href{#1}{#2}}
\providecommand{\doi}[1]{doi:\urlalt{http://dx.doi.org/#1}{#1}}
\providecommand{\bibinfo}[2]{#2}

\bibitemdeclare{inproceedings}{Bergmann14a}
\bibitem{Bergmann14a}
\bibinfo{author}{G{\'{a}}bor \surnamestart Bergmann\surnameend}
  (\bibinfo{year}{2014}): \emph{\bibinfo{title}{Translating {OCL} to Graph
  Patterns}}.
\newblock In: {\sl \bibinfo{booktitle}{Model-Driven Engineering Languages and
  Systems ({MODELS} 2014)}}, \bibinfo{series}{LNCS}, pp.
  \bibinfo{pages}{670--686}, \doi{10.1007/978-3-319-11653-2\_41}.

\bibitemdeclare{article}{Biermann-Ermel-Taetzer12a}
\bibitem{Biermann-Ermel-Taetzer12a}
\bibinfo{author}{Enrico \surnamestart Biermann\surnameend},
  \bibinfo{author}{Claudia \surnamestart Ermel\surnameend} \&
  \bibinfo{author}{Gabriele \surnamestart Taentzer\surnameend}
  (\bibinfo{year}{2012}): \emph{\bibinfo{title}{Formal foundation of consistent
  EMF model transformations by algebraic graph transformation}}.
\newblock {\sl \bibinfo{journal}{Software and System Modeling}}
  \bibinfo{volume}{11}(\bibinfo{number}{2}), pp. \bibinfo{pages}{227--250},
  \doi{10.1007/s10270-011-0199-7}.

\bibitemdeclare{inproceedings}{Cheng-etal18a}
\bibitem{Cheng-etal18a}
\bibinfo{author}{Yurong \surnamestart Cheng\surnameend}, \bibinfo{author}{Lei
  \surnamestart Chen\surnameend}, \bibinfo{author}{Ye~\surnamestart
  Yuan\surnameend} \& \bibinfo{author}{Guoren \surnamestart Wang\surnameend}
  (\bibinfo{year}{2018}): \emph{\bibinfo{title}{Rule-Based Graph Repairing:
  Semantic and Efficient Repairing Methods}}.
\newblock In: {\sl \bibinfo{booktitle}{34th {IEEE} International Conference on
  Data Engineering, {ICDE} 2018,}}, pp. \bibinfo{pages}{773--784},
  \doi{10.1109/ICDE.2018.00075}.

\bibitemdeclare{book}{Ehrig-Ehrig-Prange-Taentzer06b}
\bibitem{Ehrig-Ehrig-Prange-Taentzer06b}
\bibinfo{author}{Hartmut \surnamestart Ehrig\surnameend},
  \bibinfo{author}{Karsten \surnamestart Ehrig\surnameend},
  \bibinfo{author}{Ulrike \surnamestart Prange\surnameend} \&
  \bibinfo{author}{Gabriele \surnamestart Taentzer\surnameend}
  (\bibinfo{year}{2006}): \emph{\bibinfo{title}{Fundamentals of Algebraic Graph
  Transformation}}.
\newblock \bibinfo{series}{EATCS Monographs of Theoretical Computer Science},
  \bibinfo{publisher}{Springer}.

\bibitemdeclare{book}{Ehrig-etal15a}
\bibitem{Ehrig-etal15a}
\bibinfo{author}{Hartmut \surnamestart Ehrig\surnameend},
  \bibinfo{author}{Claudia \surnamestart Ermel\surnameend},
  \bibinfo{author}{Ulrike \surnamestart Golas\surnameend} \&
  \bibinfo{author}{Frank \surnamestart Hermann\surnameend}
  (\bibinfo{year}{2015}): \emph{\bibinfo{title}{Graph and Model Transformation
  - General Framework and Applications}}.
\newblock \bibinfo{series}{Monographs in Theoretical Computer Science},
  \bibinfo{publisher}{Springer}, \doi{10.1007/978-3-662-47980-3}.

\bibitemdeclare{article}{Habel-Pennemann09a}
\bibitem{Habel-Pennemann09a}
\bibinfo{author}{Annegret \surnamestart Habel\surnameend} \&
  \bibinfo{author}{Karl-Heinz \surnamestart Pennemann\surnameend}
  (\bibinfo{year}{2009}): \emph{\bibinfo{title}{Correctness of High-Level
  Transformation Systems Relative to Nested Conditions}}.
\newblock {\sl \bibinfo{journal}{Mathematical Structures in Computer Science}}
  \bibinfo{volume}{19}, pp. \bibinfo{pages}{245--296},
  \doi{10.1017/S0960129500001353}.

\bibitemdeclare{inproceedings}{Habel-Plump01a}
\bibitem{Habel-Plump01a}
\bibinfo{author}{Annegret \surnamestart Habel\surnameend} \&
  \bibinfo{author}{Detlef \surnamestart Plump\surnameend}
  (\bibinfo{year}{2001}): \emph{\bibinfo{title}{Computational Completeness of
  Programming Languages Based on Graph Transformation}}.
\newblock In: {\sl \bibinfo{booktitle}{Foundations of Software Science and
  Computation Structures (FOSSACS 2001)}}, {\sl \bibinfo{series}{Lecture Notes
  in Computer Science}} \bibinfo{volume}{2030}, pp. \bibinfo{pages}{230--245},
  \doi{10.1007/BFb0017401}.

\bibitemdeclare{inproceedings}{Habel-Sandmann18a}
\bibitem{Habel-Sandmann18a}
\bibinfo{author}{Annegret \surnamestart Habel\surnameend} \&
  \bibinfo{author}{Christian \surnamestart Sandmann\surnameend}
  (\bibinfo{year}{2018}): \emph{\bibinfo{title}{Graph Repair by Graph
  Programs}}.
\newblock In: {\sl \bibinfo{booktitle}{Graph Computation Models ({GCM} 2018)}},
  {\sl \bibinfo{series}{Lecture Notes in Computer Science}}
  \bibinfo{volume}{11176}, pp. \bibinfo{pages}{431--446},
  \doi{10.1007/s10009-018-0496-3}.

\bibitemdeclare{article}{Loewe93a}
\bibitem{Loewe93a}
\bibinfo{author}{Michael \surnamestart L{\"o}we\surnameend}
  (\bibinfo{year}{1993}): \emph{\bibinfo{title}{Algebraic Approach to
  Single-Pushout Graph Transformation}}.
\newblock {\sl \bibinfo{journal}{Theoretical Computer Science}}
  \bibinfo{volume}{109}, pp. \bibinfo{pages}{181--224},
  \doi{10.1016/0304-3975(93)90068-5}.

\bibitemdeclare{article}{Macedo-etal17a}
\bibitem{Macedo-etal17a}
\bibinfo{author}{Nuno \surnamestart Macedo\surnameend}, \bibinfo{author}{Jorge
  \surnamestart Tiago\surnameend} \& \bibinfo{author}{Alcino \surnamestart
  Cunha\surnameend} (\bibinfo{year}{2017}): \emph{\bibinfo{title}{A
  Feature-Based Classification of Model Repair Approaches}}.
\newblock {\sl \bibinfo{journal}{{IEEE} Trans. Software Eng.}}
  \bibinfo{volume}{43}(\bibinfo{number}{7}), pp. \bibinfo{pages}{615--640},
  \doi{10.1109/TSE.2016.2620145}.

\bibitemdeclare{inproceedings}{Nebras-etal17a}
\bibitem{Nebras-etal17a}
\bibinfo{author}{Nebras \surnamestart Nassar\surnameend},
  \bibinfo{author}{Hendrik \surnamestart Radke\surnameend} \&
  \bibinfo{author}{Thorsten \surnamestart Arendt\surnameend}
  (\bibinfo{year}{2017}): \emph{\bibinfo{title}{Rule-Based Repair of {EMF}
  Models: An Automated Interactive Approach}}.
\newblock In: {\sl \bibinfo{booktitle}{Theory and Practice of Model
  Transformation ({ICMT} 2017)}}, {\sl \bibinfo{series}{Lecture Notes in
  Computer Science}} \bibinfo{volume}{10374}, pp. \bibinfo{pages}{171--181},
  \doi{10.1007/978-3-319-21145-9\_10}.

\bibitemdeclare{inproceedings}{Nentwich-etal03a}
\bibitem{Nentwich-etal03a}
\bibinfo{author}{Christian \surnamestart Nentwich\surnameend},
  \bibinfo{author}{Wolfgang \surnamestart Emmerich\surnameend} \&
  \bibinfo{author}{Anthony \surnamestart Finkelstein\surnameend}
  (\bibinfo{year}{2003}): \emph{\bibinfo{title}{Consistency Management with
  Repair Actions}}.
\newblock In: {\sl \bibinfo{booktitle}{Software Engineering}},
  \bibinfo{publisher}{{IEEE} Computer Society}, pp. \bibinfo{pages}{455--464},
  \doi{10.1109/ICSE.2003.1201223}.

\bibitemdeclare{misc}{web:OCL24}
\bibitem{web:OCL24}
\bibinfo{author}{\surnamestart OMG\surnameend}: \emph{\bibinfo{title}{{Object
  Constraint Language}}}.
\newblock \bibinfo{howpublished}{{\url{https://www.omg.org/spec/OCL/}}}.

\bibitemdeclare{phdthesis}{Pennemann09a}
\bibitem{Pennemann09a}
\bibinfo{author}{Karl-Heinz \surnamestart Pennemann\surnameend}
  (\bibinfo{year}{2009}): \emph{\bibinfo{title}{Development of Correct Graph
  Transformation Systems}}.
\newblock Ph.D. thesis, \bibinfo{school}{Universit\"at Oldenburg}.

\bibitemdeclare{inproceedings}{Plump05a}
\bibitem{Plump05a}
\bibinfo{author}{Detlef \surnamestart Plump\surnameend} (\bibinfo{year}{2005}):
  \emph{\bibinfo{title}{Confluence of Graph Transformation Revisited}}.
\newblock In: {\sl \bibinfo{booktitle}{Processes, Terms and Cycles: Steps on
  the Road to Infinity}}, {\sl \bibinfo{series}{Lecture Notes in Computer
  Science}} \bibinfo{volume}{3838}, pp. \bibinfo{pages}{280--308},
  \doi{10.1007/BF00289616}.

\bibitemdeclare{article}{Radke+18a}
\bibitem{Radke+18a}
\bibinfo{author}{Hendrik \surnamestart Radke\surnameend},
  \bibinfo{author}{Thorsten \surnamestart Arendt\surnameend},
  \bibinfo{author}{Jan~Steffen \surnamestart Becker\surnameend},
  \bibinfo{author}{Annegret \surnamestart Habel\surnameend} \&
  \bibinfo{author}{Grabriele \surnamestart Taentzer\surnameend}
  (\bibinfo{year}{2018}): \emph{\bibinfo{title}{Translating Essential {OCL}
  Invariants to Nested Graph Constraints for Generating nstances of
  Meta-models}}.
\newblock {\sl \bibinfo{journal}{Science of Computer Programming}}
  \bibinfo{volume}{152}, pp. \bibinfo{pages}{38--62},
  \doi{10.1016/j.scico.2017.08.006}.

\bibitemdeclare{inproceedings}{Schneider-etal19a}
\bibitem{Schneider-etal19a}
\bibinfo{author}{Sven \surnamestart Schneider\surnameend},
  \bibinfo{author}{Leen \surnamestart Lambers\surnameend} \&
  \bibinfo{author}{Fernando \surnamestart Orejas\surnameend}
  (\bibinfo{year}{2019}): \emph{\bibinfo{title}{A Logic-Based Incremental
  Approach to Graph Repair}}.
\newblock In: {\sl \bibinfo{booktitle}{Fundamental Approaches to Software
  Engineering - ({FASE} 2019)}}, {\sl \bibinfo{series}{Lecture Notes in
  Computer Science}} \bibinfo{volume}{11424}, pp. \bibinfo{pages}{151--167},
  \doi{10.1007/978-3-662-54494-5\_16}.

\bibitemdeclare{inproceedings}{Schuerr94b}
\bibitem{Schuerr94b}
\bibinfo{author}{Andy \surnamestart Sch{\"{u}}rr\surnameend}
  (\bibinfo{year}{1994}): \emph{\bibinfo{title}{Specification of Graph
  Translators with Triple Graph Grammars}}.
\newblock In: {\sl \bibinfo{booktitle}{Graph-Theoretic Concepts in Computer
  Science, 20th International Workshop, {WG} '94, Herrsching, Germany, June
  16-18, 1994, Proceedings}}, pp. \bibinfo{pages}{151--163}.
\newblock \urlprefix\url{https://doi.org/10.1007/3-540-59071-4\_45}.

\bibitemdeclare{inproceedings}{Taentzer-etal17a}
\bibitem{Taentzer-etal17a}
\bibinfo{author}{Gabriele \surnamestart Taentzer\surnameend},
  \bibinfo{author}{Manuel \surnamestart Ohrndorf\surnameend},
  \bibinfo{author}{Yngve \surnamestart Lamo\surnameend} \&
  \bibinfo{author}{Adrian \surnamestart Rutle\surnameend}
  (\bibinfo{year}{2017}): \emph{\bibinfo{title}{Change-Preserving Model
  Repair}}.
\newblock In: {\sl \bibinfo{booktitle}{Fundamental Approaches to Software
  Engineering ({ETAPS} 2017)}}, {\sl \bibinfo{series}{Lecture Notes in Computer
  Science}} \bibinfo{volume}{10202}, pp. \bibinfo{pages}{283--299},
  \doi{10.1007/11880240\_15}.

\end{thebibliography}

\end{document}